\documentclass[onecolumn,12pt,draftclsnofoot]{IEEEtran}

\usepackage{mathrsfs}
\usepackage{mathptmx}
\usepackage{amsmath}
\usepackage{amsfonts}
\usepackage{amssymb}
\usepackage{graphicx}
\usepackage{xcolor}
\usepackage{float}
\usepackage{subfigure}
\usepackage{psfrag}
\usepackage{multirow}
\usepackage{bm}
\usepackage{cite}
\usepackage{booktabs}

\newtheorem{thm}{Theorem}

\newtheorem{cor}{Corollary}
\newtheorem{rem}{Remark}

\newtheorem{lem}{Lemma}
\def\tp{\mathrm{T}}
\def\sn{\mathrm{span}}
\def\nl{\mathrm{null}}
\def\rank{\mathrm{rank}}
\def\Ds{\displaystyle}

\renewcommand\baselinestretch{1.35}

\begin{document}

\title{From Control to Mathematics--Part II: Observability-Based Design for Iterative Methods in Solving Linear Equations}


\author{Deyuan Meng, {\it Senior Member}, {\it IEEE}

\thanks{This work was supported by the National Natural Science Foundation of China under Grant 61922007 and Grant 61873013.}
\thanks{The author is with the Seventh Research Division, Beihang University (BUAA), Beijing 100191, P. R. China, and also with the School of Automation Science and Electrical Engineering, Beihang University (BUAA), Beijing 100191, P. R. China (e-mail: dymeng@buaa.edu.cn).}
}

\date{}
\maketitle

\begin{abstract}
The control approaches generally resort to the tools from the mathematics, but whether and how the mathematics can benefit from the control approaches is unclear. This paper aims to bring the ``control design'' idea into the mathematics by providing an observer-based iterative method that focuses on solving linear algebraic equations (LAEs). An inherent relationship is revealed between the problem-solving of LAEs and the design of observer-based control systems, with which the iterative method for solving LAEs is exploited based on the design of the basic state observers. It is shown that all (least squares) solutions for any (un)solvable LAEs can be determined exponentially fast or monotonically with different selections of initial conditions. By integrating the design idea of the deadbeat control, the solving of LAEs can be achieved within only finite iterations. In particular, our proposed iterative method can be leveraged to develop a new observer-based design algorithm to realize the perfect tracking objective of conventional two-dimensional iterative learning control (ILC) systems, where the gap between classical ILC design and popular feedback-based control design is narrowed.
\end{abstract}


\section{Introduction}\label{sec1}

\IEEEPARstart{C}{ontrol} design and analysis are generally implemented by resorting to the well-established methods and theories from the mathematics. This is well-known for all control fields, wherein ``... respect for mathematical rigor has been a hallmark of control systems research ...,'' as said by {\AA}str\"{o}m and Kumar \cite{ak:14}. In particular, of note is that as one of the most fundamental mathematic problems, the solving of linear algebraic equations (LAEs) is widely performed in the control design and analysis. The LAEs as well as the presented solving methods and results of them usually provide necessary design and analysis tools for the control systems, especially for those linear control systems, and there exist also plenty of popular control problems that can be equivalently transformed into the problem-solving of LAEs. However, limited attentions have been devoted to bettering the interaction between mathematics and control on the other side, i.e., to developing how control can feed back into mathematics. By contrast, there are control methods that may help to provide new perspectives into the study of some mathematic problems, such as iterative learning control (ILC) that has a close relation with iterative methods, especially for solving LAEs \cite{bta:06,acm:07,mw:21}.

\subsection{A Motivating Example: Designing ILC and Solving LAEs}\label{sec11}

As a class of intelligent control methods, ILC is particularly available for improving the transient performances of dynamic systems that run repeatedly over a finite time horizon with only measurement and control data (see, e.g., \cite{bta:06,mw:21,acm:07,s:05} and references therein). Typical ILC systems generally evolve along two axes, i.e., the time axis and the iteration axis. A common objective of ILC systems is to track some desired reference trajectory over any finite time interval of interest. The working mechanism of ILC to realize a perfect output tracking objective is the iterative update of the control input based on information from previous iterations. Through over-and-over iterations, the output finally can approach the desired reference trajectory under the impacts of the relevant input that converges to some desired input \cite{mw:21}.

When of interest is the study on ILC for discrete-time linear systems, the finite duration of them along the time axis makes it possible to incorporate their time-domain dynamics into their iteration-domain dynamics with the help of a lifting technique \cite{bta:06,acm:07}. Thanks to using supervectors resulting from the lifting technique, a linear mapping relation can be directly established between the output supervector and the input supervector along the iteration axis, for which a structured mapping matrix arises (see, e.g., \cite[eq. (16)]{acm:07}). Correspondingly, the perfect tracking problem of ILC systems becomes the iterative seeking problem of the desired input supervector subject to the desired reference supervector under the same linear mapping relation. Namely, it discloses a fact that the design problem of ILC may exactly fall into the problem-solving framework of LAEs through iterative methods. As a benefit, new ILC design algorithms may emerge from the iterative methods for the solving of LAEs, especially in contrast with the commonly used ILC algorithms that adopt a class of PID-type iterative updating laws (see detailed results in the following Section \ref{sec5} and also similar discussions in \cite{mw:21}).

As core problems in the control area, the design and analysis for controllers are naturally considered to play a fundamentally important role for ILC, which are consequently integrated into the iterative methods associated with ILC. This actually brings new perspectives into the iterative methods in the mathematics, especially into those for solving LAEs, because they generally do not involve any control design problems (see, e.g., \cite{k:95,g:97,sv:00,mc:05,f:53} and references therein). Take, for example, the P-type discrete-time linear ILC, where the synthesis of the control gain matrix helps to improve the convergence rate, achieve the disturbance rejection, and optimize the specified system performance index for the resulting iterative process of ILC \cite{s:05,bta:06}. Although ILC relates to a specific class of iterative methods, these advantages arising from the control design may provide us the motivations to further explore whether and how the control design can help to improve more general iterative methods in the mathematics, such as those in solving LAEs. In addition, salient convergence analysis methods for ILC, such as those based on contraction mapping and two-dimensional (2-D) systems \cite{acm:07}, may provide us the possible ways to develop new convergence analysis tools of the iterative methods in the mathematics.

From the aforementioned discussions on bridging a relation between the design of ILC and the solving of LAEs, there may exist interactions between control and mathematics from both sides: the mathematics provides the basis for the control, while the control feeds back into the mathematics simultaneously. It thus motivates us to make contributions to promoting the two-sided interaction between mathematics and control by focusing on the problem-solving of LAEs through the iterative methods.

\subsection{A Brief Review of Related Works}

The solving of LAEs, as one of the fundamental problems in the mathematics, is indispensable in many scientific researches and engineering applications (see also \cite{k:95,g:97,sv:00,mc:05}). There are some other mathematic problems that can be addressed by benefiting from the derived methods and results of solving LAEs, such as the solving of linear matrix equations like Lyapunov equations and Sylvester equations \cite{hj:91,h:18}. Two classes of methods for solving LAEs are usually used, namely, the direct methods and the iterative methods \cite{mc:05}. For the direct methods, the Gaussian elimination is implemented such that the solution for any LAE can be directly gained after a calculation process of finite steps. This is particularly effective for LAEs with lower dimensions, but when considering LAEs with higher dimensions, the direct methods are usually subjected to huge computation and storage costs and not practical any longer. By comparison, the iterative methods are more preferred in the practical applications thanks to their simple and easy-to-implement processes that only need addition and multiplication operations among several matrices and vectors. Further, owing to the issue of calculation accuracy (or robustness \cite{bk:03,bk:07}) in the presence of the round-off error, ``... direct methods disappear except for quite simple problems, and all methods become iterative ...,'' as said by Forsythe \cite{f:53}.

For the iterative methods, the implementation mechanism of solving LAEs is to gradually update the estimated solutions for LAEs at each iteration, which is generally realized by resorting to the estimation information only from the previous iterations. There have been introduced various effective iterative methods that are aimed at the solving of LAEs, such as the Gauss-Seidel method, the Krylov subspace method, and the successive over-relaxation method. For a more detailed literature review of the iterative methods for solving LAEs, we refer the readers also to \cite{k:95,g:97,sv:00,mc:05,f:53}. For the classic iterative methods, one of the desirable objectives is to ensure the asymptotic convergence to the exact solutions of LAEs. However, even though they ``feed back'' the previous estimation results into the calculation of the estimated solutions for LAEs at the current iteration, the classic iterative methods are subjected to difficulties to ``control'' their updating processes, where the convergence of them may not always hold and the convergence speed of them may be unclear. A question naturally emerges: whether and how can the design of iterative methods be improved, ensuring them to operate in an effective and systematic regulation that can be controlled as desired?

To find affirmative answers to the abovementioned question, we may remind ourselves of the key idea of automatic control, which is to give us extra degrees of freedom to make controlled systems possible to achieve certain desired requirements, such as those of stability, accuracy, response speed, and robustness \cite{ak:14}. This observation provides motivation for the incorporation of the insightful idea of ``control design'' into the design of the iterative methods to promote their performances especially for solving LAEs, based on which there have been reproted several instructive and valuable attempts in the literature \cite{bk:03,bk:07,hjl:06,hj:05,yst:16,ae:20,mlm:15,lmnb:17,ygqyw:20}. In \cite{bk:03,bk:07}, the solving problem of LAEs has been transformed into the stability problem of constructed residue systems, under which the feedback control design can be successfully adopted for the execution of iterative methods, together with taking full advantage of the famous Lyapunov stability analysis theory. To avoid exhibiting the unstable dynamics or remove the sensitive dependence on the initial conditions, the idea from optimal and robust control has been incorporated into the iterative methods for solving LAEs in \cite{hjl:06,hj:05}, which can guarantee the global convergence of them to the desired solutions of LAEs in spite of any initial states. In addition, a lot of effort has been recently spent to integrate the design mechanism of distributed control for multi-agent systems into the iterative solution methods for LAEs \cite{yst:16,ae:20,mlm:15,lmnb:17,ygqyw:20}, which may help to lighten the computational burden and improve the computational efficiency (for detailed discussions, see, e.g., \cite{wmlr:19}).

Despite the notable attempts made in \cite{bk:03,bk:07,hjl:06,hj:05,yst:16,ae:20,mlm:15,lmnb:17,ygqyw:20}, the methods and results for how to incorporate the control-theoretic design ideas into the iterative methods are still less developed, which has already been disclosed and highlighted for the interaction between control and numerical analysis \cite{c:01}. Of particular note is that the existing control-theoretic design methods for solving LAEs are subjected to some specific constraints on LAEs. For example, the basic solvability hypothesis on LAEs is required in \cite{bk:03,bk:07,hjl:06,hj:05,yst:16,ae:20,mlm:15,lmnb:17}, where the uniqueness assumption is also imposed for the solutions to LAEs in \cite{bk:03,bk:07,hjl:06,hj:05,yst:16,ae:20}. Although an attempt to find the least squares solutions for unsolvable LAEs is made in \cite{ygqyw:20}, a basic full-column rank condition is needed to guarantee the uniqueness of the least squares solutions for LAEs. In fact, it is unclear whether there exist some general control-theoretic design ideas for the iterative methods that can effectively work, regardless of any solvable or unsolvable LAEs. This is because it is lack of methods and results to reasonably connect the basic properties of control, such as controllability and  observability, to the (least squares) solutions for LAEs. Further, by recalling the basic hypotheses required for the control design \cite{am:06,r:96}, it is easy to see that such a reasonable connection may provide a fundamental guideline for the control-theoretic design of the iterative methods, like controllability for a feedback controller design and observability for a state observer design. However, it still remains to be explored whether and how these observed issues on control-based iterative methods can be addressed.

In addition, when concerning LAEs that have multiple (least squares) solutions, it may be desirable to find design tools for the iterative methods such that they are capable of determining all possible (least squares) solutions of LAEs. To proceed, it is expected to establish how the initial conditions of the iterative methods contribute to determining the multiple (least squares) solutions of LAEs, which consequently helps to obtain specific (least squares) solutions by appropriate selections of the initial conditions. Generally, the asymptotic/exponential convergence can only be guaranteed for the existing control-based iterative methods \cite{bk:03,bk:07,hjl:06,hj:05,yst:16,ae:20,mlm:15,lmnb:17,ygqyw:20}, and thus it is left to further develop how to improve the convergence speeds of them, especially by making full use of the control design methods.

\subsection{Our Contributions and Paper Organization}

\begin{figure}
\centering
\includegraphics[width=3in]{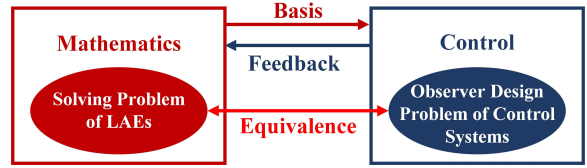}
\caption{An illustration of the relation between mathematics and control.}\label{fig1}
\end{figure}

In this paper, we contribute to bettering the existing iterative methods of solving LAEs, and provide them with a new design perspective by taking full advantage of the design idea for state observers in linear systems and control \cite{am:06,r:96}. We propose a control-theoretic design framework for the iterative methods, which incorporates any given LAE of interest into the dynamic process of a discrete linear system and, thus, makes it possible to reveal the equivalence relation between the solving problem of an LAE and the state observer design problem of a resulting linear system. Clearly, this provides a way to feed back control design into the mathematics, and can promote the bidirectional interaction between mathematics and control, as illustrated by Fig. \ref{fig1}. The main contributions for our control-theoretic design approach to the iterative methods are summarized as follows.
\begin{enumerate}
\item
A tight connection is reasonably established between the fundamental observability properties for control systems and the (least squares) solutions for LAEs, which is due to the incorporation of the design idea for state observers into the problem-solving of LAEs. It is also revealed that for any LAE, together with the linear system associated with it, the general solutions of the LAE exactly span the unobservable subspace of its associated linear system. In contrast to this, the particular (least squares) solutions of the LAE have close connections to the observable states of its associated linear system. For clarity, an illustration of these facts is provided in Fig. \ref{fig2}.

\begin{figure}
\centering
\includegraphics[width=4in]{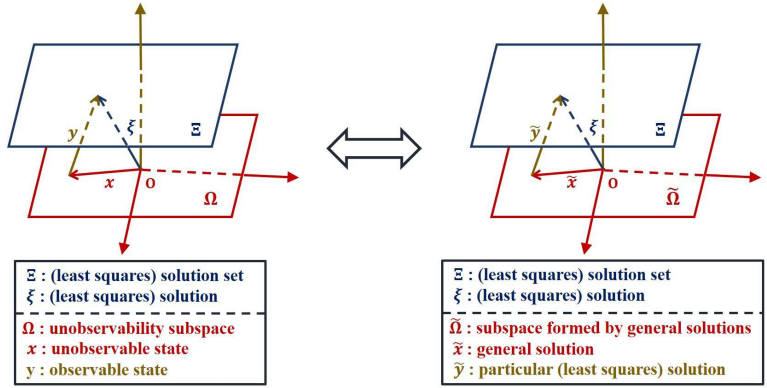}
\caption{The additive decomposition of any (least squares) solution to an LAE from the observerability perspective.}\label{fig2}
\end{figure}

\item
A new iterative solution algorithm that involves a control gain matrix (i.e., an observer gain matrix) to be designed is proposed for any LAE, regardless of whether the LAE is solvable or not, and whether it admits a unique (least squares) solution or multiple (least squares) solutions. It is shown that for any solvable (respectively, unsolvable) LAEs, all solutions (respectively, least squares solutions) of them can be determined in an analytical form. Further, for any LAE, a deterministic relation between the initial conditions of our iterative solution algorithm and all the (least squares) solutions of it can be developed.

\item
Conditions for not only the exponential convergence but also the monotonic convergence of our iterative solution algorithm are presented in the presence of any LAEs. To further improve its convergence speed, the design idea of deadbeat control is leveraged to establish conditions that can reach the finite-iteration convergence of our iterative solution algorithm. Our design approaches to monotonic convergence and finite-iteration convergence may help to provide new ways and ideas for bettering the exponential convergence results derived by the existing control-based iterative methods for solving LAEs (see, e.g., \cite{bk:03,bk:07,hjl:06,hj:05,yst:16,ae:20,mlm:15,lmnb:17,ygqyw:20}).

\item
Our iterative method for solving LAEs is applied to give a new design and analysis framework of the conventional 2-D ILC systems, under which the perfect tracking tasks are realized through the direct seeking of desired inputs. This leads to ILC algorithms that fall beyond the typical ILC design framework. In particular, the perfect tracking of ILC can be achieved for linear systems, even without requiring the full rank condition about their first nonzero Markov parameter matrices, which is generally regarded as a basic condition for the implementation of ILC (see, e.g., \cite{sw:02,xypy:16,mm:171,szwc:16}). As a consequence, the relation between the ILC design and the popular feedback-based control design is bridged. Besides, by certain specific selections of the gain matrix, the perfect tracking objectives of ILC can be accomplished only within finite-iteration steps.
\end{enumerate}

\noindent From these aspects, the interaction between designing ILC and solving LAEs is also verified, as stated in the Subsection \ref{sec11}.

We organize the remaining sections of our paper as follows. In Section \ref{sec2}, we introduce the problem description for solving LAEs, of which the equivalence relation with the state observer design problem for linear systems is established in Section \ref{sec7}. Despite any LAEs, we develop the design and analysis results of our proposed iterative solution algorithm in Section \ref{sec3}, and moreover provide it with a finite-iteration convergence result in Section \ref{sec4}. We consider the application of our iterative method for solving LAEs to the conventional 2-D ILC systems design in Section \ref{sec5}, and then make some concluding remarks about its validity and effectiveness in Section \ref{sec6}.

Before ending this section, we introduce some notations that are needed in the following discussions.

{\it Notations:} Let $\mathbb{Z}_{+}=\{0, 1, 2, \cdots\}$ be the set of nonnegative integers. Denote $\mathbb{Z}_N=\{0, 1, \cdots, N \}$ for $N\in\mathbb{Z}_{+}$ and $N\neq0$, and $I_{n}\in\mathbb{R}^{n\times n}$ as the identity matrix. For any $\bm{x}\in\mathbb{R}^{n}$, $\|\bm{x}\|$ is its any vector norm, and $\|\bm{x}\|_{2}$ is its specific Euclidean norm. For any $\bm{A}\in\mathbb{R}^{m\times n}$, $\|\bm{A}\|$ is its any matrix norm, where particularly $\|\bm{A}\|_{2}$ and $\|\bm{A}\|_{F}$ denote its $2$-norm and Frobenius norm, respectively. We define the generalized inverse matrix set $\bm{A}\{1,3\}$ of $\bm{A}$ as
\[\bm{A}\{1,3\}=\left\{\bm{X}\in\mathbb{R}^{n\times m}\big|\bm{A}\bm{X}\bm{A}=\bm{A},\left(\bm{A}\bm{X}\right)^{\tp}=\bm{A}\bm{X}\right\}
\]

\noindent where any $\bm{X}\in\bm{A}\{1,3\}$ is called a $\{1,3\}$-inverse of $\bm{A}$, and also denoted directly by $\bm{A}^{(1,3)}$ (see Definition 1 of \cite[Chapter 1]{bg:03}). In the case of any square matrix $\bm{A}\in\mathbb{R}^{n\times n}$, its spectral radius is denoted as $\rho(\bm{A})=\max_{i=1,2,\cdots,n}\left\{\left|\lambda_{i}(\bm{A})\right|\right\}$, where $\lambda_{i}(\bm{A})$ denotes any (or the $i$th) eigenvalue of $\bm{A}$. If $\lambda_{i}(\bm{A})\in\mathbb{R}$, $\forall i=1$, $2$, $\cdots$, $n$, then we denote $\lambda_{\max}(\bm{A})=\max_{i=1,2,\cdots,n}\left\{\lambda_{i}(\bm{A})\right\}$. For any set $\bm{\Omega}\subseteq\mathbb{R}^{n}$ and $\bm{x}\in\mathbb{R}^{n}$, let $\bm{\Omega}+\bm{x}=\left\{\bm{x}+\bm{y}\big|\bm{y}\in\bm{\Omega}\right\}$, and denote $\bm{\Omega}^{\perp}$ as the orthogonal complement subspace of $\bm{\Omega}$ when it further constitutes a linear subspace in $\mathbb{R}^{n}$.

\section{Problem Description of Solving LAEs}\label{sec2}

Consider an LAE with any nonzero mapping/transfer matrix $\bm{G}\in\mathbb{R}^{p\times q}$ and any nonzero vector $\bm{Y}_{d}\in\mathbb{R}^{p}$ such that
\begin{equation}\label{eq1}
\bm{Y}_{d}=\bm{G}\bm{U}_{d}~~\hbox{for}~~\bm{U}_{d}\in\mathbb{R}^{q}.
\end{equation}

\noindent For any given $\bm{Y}_{d}\in\mathbb{R}^{p}$, our objective is to establish an iterative method for the solving of the LAE (\ref{eq1}), based on which we can obtain the (least squares) solution $\bm{U}_{d}\in\mathbb{R}^{q}$. If there exist some (respectively, no) solutions for the LAE (\ref{eq1}), then we say that it is {\it solvable} (respectively, {\it unsolvable}). Thus, we further target at developing an analytical description as a linear function of the initial condition for all the solutions (respectively, least squares solutions), regardless of any solvable (respectively, unsolvable) LAE (\ref{eq1}). 

Though there are many effective iterative methods given for solving LAEs, they may either impose some constraints on the LAEs or be incapable of deriving the analytical description of all (least squares) solutions for LAEs (see, e.g.,
\cite{bk:03,bk:07,hjl:06,hj:05,yst:16,ae:20,mlm:15,lmnb:17,ygqyw:20}). In contrast, the idea of ``control design'' that is exploited based on the methods from the mathematics may help to overcome these drawbacks by introducing the design freedom and may, in turn, provide an effective and systematic way to promote particular investigations of the mathematics. We therefore plan to explore a new iterative method that incorporates the control design for the problem-solving of the LAE (\ref{eq1}), and moreover develop the interaction between the control and the mathematics. There are generally two dual approaches to achieving the control design, namely, controller design and observer design, where {\it we focus on using the observer design in the current paper to show how the control design betters iterative methods in solving LAEs.}

\section{An Observer Perspective for Solving LAEs}\label{sec7}

In this section, we first implement a problem transformation for the solving of LAEs by bridging the connection between it and the state observer design of linear systems, based on which we then contribute to establishing the relationship between the solving of LAEs and the observability of linear systems.

\subsection{Problem Transformation: An Observer Design Perspective}

To accomplish the abovementioned objectives of solving the LAE (\ref{eq1}), we provide a new perspective into iterative methods from the observer design of control systems. By this viewpoint, we start with the hypothesis of the solvability for the LAE (\ref{eq1}), and can include (\ref{eq1}) as the trivial steady-state case of the output equation of a discrete control system. To be specific, let $k\in\mathbb{Z}_{+}$ be an iteration index, and denote $\bm{U}_{d}(k)\equiv\bm{U}_{d}$ and  $\bm{Y}_{d}(k)\equiv\bm{Y}_{d}$ for all $k\in\mathbb{Z}_{+}$ such that we can reformulate (\ref{eq1}) in the form of a discrete control system as
\begin{equation}\label{eq2}
\left\{\aligned
\bm{U}_{d}(k+1)&=\bm{U}_{d}(k)+\bm{V}_{d}(k)\\
\bm{Y}_{d}(k)&=\bm{G}\bm{U}_{d}(k)
\endaligned\right.,\quad\forall k\in\mathbb{Z}_{+}
\end{equation}

\noindent where $\bm{V}_{d}(k)\equiv0$, $\forall k\in\mathbb{Z}_{+}$ is just the trivial zero control input. Clearly, $\bm{U}_d(k)$ and $\bm{Y}_d(k)$ are exactly the state and the output of the system \eqref{eq2}, respectively, and the structure of this system is identified with known matrices. Of special note is the design objective of state observers for the system (\ref{eq2}), focusing on how to acquire the state $\bm{U}_{d}(k)$ (or exactly $\bm{U}_{d}$) in the sense of the asymptotic equivalence \cite{am:06}, which coincides with the solving objective of the LAE (\ref{eq1}) for iterative methods. Thanks to this observation, {\it the solving problem of the LAE (\ref{eq1}) can therefore be transformed into a well-known observability problem in the systems and control: whether and how can the system state of the system (\ref{eq2}) be determined in the presence of the knowledge of the measured output and the known system structure?}

To proceed with the discussion on the system (\ref{eq2}), we denote $\bm{U}_{k}$ as the ``state estimation'' at the iteration $k$ for $\bm{U}_{d}$ that fulfills the LAE (\ref{eq1}) for the given $\bm{Y}_{d}$. Hence, our objective can further turn to the design of appropriate iterative algorithms such that the resulting sequence $\left\{\bm{U}_{k}:k\in\mathbb{Z}_{+}\right\}$ converges to the solution of the solvable LAE (\ref{eq1}). Motivated by these facts, we may then employ the design ideas of ``state observers'' for linear systems to achieve the objective of solving LAEs.

Note that the aforementioned discussions are made based on the hypothesis of the solvability for the LAE (\ref{eq1}). For the other case that the LAE (\ref{eq1}) is unsolvable, $\bm{Y}_{d}-\bm{G}\widetilde{\bm{U}}_{d}\neq0$, $\forall\widetilde{\bm{U}}_{d}\in\mathbb{R}^{q}$ holds for the given $\bm{Y}_{d}\in\mathbb{R}^{p}$. Hence, the system (\ref{eq2}) under any initial condition $\bm{U}_{d}(0)=\bm{U}_{d}$ becomes
\begin{equation}\label{eq59}
\left\{\aligned
\bm{U}_{d}(k+1)&=\bm{U}_{d}(k)+\bm{V}_{d}(k)\\
\widetilde{\bm{Y}}_{d}(k)&=\bm{G}\bm{U}_{d}(k)
\endaligned\right.,\quad\forall k\in\mathbb{Z}_{+}
\end{equation}

\noindent where $\widetilde{\bm{Y}}_{d}(k)\equiv\widetilde{\bm{Y}}_{d}=\bm{G}\bm{U}_{d}$, $\forall k\in\mathbb{Z}_{+}$. This system clearly shares the same structure as the system (\ref{eq2}), but is subject to a constant steady-state output error $\bm{Y}_{d}-\bm{G}\bm{U}_{d}\neq0$. We may consequently expect to validate that the design idea of $\left\{\bm{U}_{k}:k\in\mathbb{Z}_{+}\right\}$ based on the state observers still works effectively in the presence of an unsolvable LAE (\ref{eq1}), and further modify the design objective of this sequence to guarantee it to converge to the least squares solution of (\ref{eq1}), which thus minimizes the resulting steady-state output error in the sense of the Euclidean norm.

\subsection{Observability and the Solving of LAEs}

In this subsection, we disclose how the observability affects the solving of LAEs, based on which we use the design method of state observers to gain an observer-based iterative algorithm for solving the LAE (\ref{eq1}). For convenience, we refer the readers directly to, e.g., \cite{am:06} for concepts related with the observability of linear systems. We also simply call the observability instead of the complete (state) observability if no confusions can arise.

We consider the system (\ref{eq2}) that is constructed from the LAE (\ref{eq1}), and denote $\mathcal{O}_{NO}$ as its {\it unobservability subspace} according to \cite[Definition 3.13]{am:06}. For clarity, we also denote $\mathcal{O}_{O}=\mathcal{O}_{NO}^{\perp}$ as the {\it observability subspace} of the system (\ref{eq2}). We can hence present the following properties of $\mathcal{O}_{NO}$ and $\mathcal{O}_{O}$.

\begin{lem}\label{lem6}
For the system (\ref{eq2}), it holds:
\begin{enumerate}
\item
$\mathcal{O}_{NO}=\rm{null}\bm{G}$;

\item
$\mathcal{O}_{O}=\rm{span}\bm{G}^{\tp}$.
\end{enumerate}
\end{lem}

\begin{proof}
{\it1):} A consequence of applying \cite[Theorem 3.11]{am:06} directly to the system (\ref{eq2}).

{\it2):} A counterpart result of 1) by noting two facts that $\mathcal{O}_{O}=\mathcal{O}_{NO}^{\perp}$ and $q=\rm{dim}\left(\rm{null}\bm{G}\right)+\rm{dim}\left(\rm{span}\bm{G}^{\tp}\right)$.
\end{proof}

Since $\rm{null}\bm{G}$ is exactly the solution space of the corresponding homogeneous equation $\bm{G}\bm{U}_{d}=0$ for the LAE (\ref{eq1}), Lemma \ref{lem6} discloses actually that $\mathcal{O}_{NO}$ consisting of all the unobservable states of the system (\ref{eq2}) coincides with this solution space. As a counterpart, we will reveal that $\mathcal{O}_{O}$ is connected closely to the particular solutions for the LAE (\ref{eq1}). Thanks to $\mathcal{O}_{O}=\rm{span}\bm{G}^{\tp}$, the system (\ref{eq2}) is observable if and only if ${\rm{span}}\bm{G}^{\tp}=\mathbb{R}^{q}$ holds. We can thus give equivalent conditions on the observability of the system (\ref{eq2}) in the following lemma.

\begin{lem}\label{lem7}
The system (\ref{eq2}) is observable if and only if any of the following statements holds.
\begin{enumerate}
%
\item
The system (\ref{eq2}) is detectable.

\item
There exists some matrix $\bm{F}\in\mathbb{R}^{q\times p}$ such that
\begin{equation}\label{eq5}
\rho\left(I_{q}-\bm{F}\bm{G}\right)<1.
\end{equation}

\item
The matrix $\bm{G}$ has full-column rank, i.e., $\rank\left(\bm{G}\right)=q$.
\end{enumerate}
\end{lem}

\begin{proof}
Three equivalent conditions follow easily from the standard linear system theory and matrix theory (see, e.g.,\cite{am:06,lt:85}) by Lemma \ref{lem6} and due to the specific structure of (\ref{eq2}).
\end{proof}

From Lemma \ref{lem7}, we know that owing to the specific structure for the system (\ref{eq2}), the observability and the detectability of it are equivalent, and correspond to the full-column rank case of $\bm{G}$ for the LAE (\ref{eq1}). To proceed further with Lemma \ref{lem7}, we show the specific solution of the LAE \eqref{eq1} under the observability of the system (\ref{eq2}).

\begin{lem}\label{lem8}
For any solvable (respectively, unsolvable) LAE (\ref{eq1}), there exists a unique solution (respectively, a unique least squares solution), which is determined by
\[\bm{U}_{d}
=\left(\bm{G}^{\tp}\bm{G}\right)^{-1}\bm{G}^{\tp}\bm{Y}_{d}\]

\noindent if and only if the system (\ref{eq2}) is observable; and there exist multiple solutions (respectively, multiple least squares solutions), otherwise.
\end{lem}
%

\begin{proof}
If the LAE (\ref{eq1}) is solvable, then with Lemma \ref{lem7}, this proof is a consequence of considering \cite[Corollary 3.14]{am:06} for the system (\ref{eq2}). Otherwise, it follows by applying the results of Lemma \ref{lem7}, Exercise 1 of \cite[Subsection 12.8]{lt:85}, and Theorem 1 of \cite[Subsection 12.9]{lt:85}.
\end{proof}

From Lemma \ref{lem8}, the observability of the system (\ref{eq2}) provides a necessary and sufficient guarantee for the uniqueness of the solution to the solvable LAE (\ref{eq1}). It is generally a basic premise for the design of state observers for the system (\ref{eq2}). Inspired by this connection, we may pose an iterative method in the design framework of state observers to arrive at the exact solutions of solvable LAEs. Otherwise, if the system (\ref{eq2}) is not completely observable, then according to Lemma \ref{lem8}, the solvable LAE (\ref{eq1}) has multiple solutions, which can be given in the form of (see, e.g., Theorem 2 of \cite[Subsection 3.10]{lt:85})
\[
\bm{U}_{d}=\bm{U}_{d}^{\prime}+\bm{U}_{d}^{\prime\prime}
\]

\noindent where $\bm{U}_{d}^{\prime}$ is a particular solution of (\ref{eq1}) and $\bm{U}_{d}^{\prime\prime}$ is some vector belonging to $\rm{null}\bm{G}$. In this case, Lemma \ref{lem6} helps to distinguish the observable and unobservable states of the system (\ref{eq2}). Since $\bm{U}_{d}^{\prime\prime}$ corresponds to the unobservable state based on Lemma \ref{lem6}, it is only left to determine a particular solution $\bm{U}_{d}^{\prime}$ in order to derive the solution $\bm{U}_{d}$ for the LAE (\ref{eq1}), for which the observer-based iterative method may still work effectively once we only focus on the observable states of the system (\ref{eq2}). Similarly, for the unsolvable LAEs, we may try to explore the observer-based iterative method to obtain the least squares solutions for them. 

Thanks to the aforementioned observations, we consider the system (\ref{eq2}), and then can leverage the idea of designing its state observer to arrive at an iterative algorithm as
\begin{equation}\label{eq4}
\aligned
\bm{U}_{k+1}
&=\bm{U}_{k}+\bm{V}_{d}(k)+\bm{F}\left(\bm{Y}_{d}-\bm{Y}_{k}\right)\\
&=\left(I_{q}-\bm{F}\bm{G}\right)\bm{U}_{k}+\bm{F}\bm{Y}_{d},\quad\forall k\in\mathbb{Z}_{+}
\endaligned
\end{equation}

\noindent where $\bm{F}\in\mathbb{R}^{q\times p}$ is an observer gain matrix to be designed, and $\bm{Y}_{k}$ denotes the estimation of $\bm{Y}_{d}$ at the $k$th iteration satisfying 
\begin{equation}\label{eq3}
\bm{Y}_{k}=\bm{G}\bm{U}_{k},\quad\forall k\in\mathbb{Z}_{+}.
\end{equation}

\noindent Of particular note is that (\ref{eq4}) and (\ref{eq3}) form a linear system with the same structure as the system (\ref{eq2}), except the ``feedback term $\bm{F}\left(\bm{Y}_{d}-\bm{Y}_{k}\right)$.'' This leads to a direct result in the lemma below.

\begin{lem}\label{lem1}
For any LAE (\ref{eq1}), the related system (\ref{eq2}) and that formed by (\ref{eq4}) and (\ref{eq3}) share the same observability properties. In particular, for any nonsingular matrix $\bm{P}\in\mathbb{R}^{q\times q}$, the use of $\overline{\bm{U}}_{d}(k)=\bm{P}^{-1}\bm{U}_{d}(k)$ yields an equivalent internal representation of the system (\ref{eq2}) as
\begin{equation*}\label{}
\left\{\aligned
\overline{\bm{U}}_{d}(k+1)&=\overline{\bm{U}}_{d}(k)+\bm{P}^{-1}\bm{V}_{d}(k)\\
\bm{Y}_{d}(k)&=\overline{\bm{G}}\,\overline{\bm{U}}_{d}(k)
\endaligned\right.,\quad\forall k\in\mathbb{Z}_{+}
\end{equation*}

\noindent and the use of $\overline{\bm{U}}_{k}=\bm{P}^{-1}\bm{U}_{k}$ for the system formed by (\ref{eq4}) and (\ref{eq3}) leads to an equivalent internal representation of it as
\begin{equation*}\label{}
\left\{\aligned
\overline{\bm{U}}_{k+1}&=\left(I_{q}-\overline{\bm{F}}\,\overline{\bm{G}}\right)\overline{\bm{U}}_{k}+\overline{\bm{F}}\bm{Y}_{d}\\
\bm{Y}_{k}&=\overline{\bm{G}}\,\overline{\bm{U}}_{k}
\endaligned\right.,\quad\forall k\in\mathbb{Z}_{+}
\end{equation*}

\noindent where $\overline{\bm{F}}=\bm{P}^{-1}\bm{F}$ and $\overline{\bm{G}}=\bm{G}\bm{P}$.
\end{lem}

We will reveal that the property of Lemma \ref{lem1} in fact ensures the iterative algorithm formed by  (\ref{eq4}) and (\ref{eq3}) works for solving the LAE (\ref{eq1}), regardless of any matrix $\bm{G}$ and any given vector $\bm{Y}_{d}$, where the design of the feedback term plays a key role.

\begin{rem}\label{rem10}
Need of special note is the design of the iterative algorithm (\ref{eq4}) for the unsolvable LAE (\ref{eq1}) to determine the least squares solutions. Since $\bm{U}_{d}$ is a least squares solution of (\ref{eq1}) if and only if it is a solution of the resulting LAE from (\ref{eq1}) such that $\bm{G}\bm{G}^{(1,3)}\bm{Y}_{d}=\bm{G}\bm{U}_{d}$, $\forall\bm{G}^{(1,3)}\in\bm{G}\{1,3\}$ (see, e.g., Corollary 1 of \cite[Chapter 3]{bg:03}), we revisit (\ref{eq59}) by taking $\widetilde{\bm{Y}}_{d}=\bm{G}\bm{G}^{(1,3)}\bm{Y}_{d}$, and then can leverage the state observer design of (\ref{eq59}) to present an iterative algorithm devoted to approaching the least squares solutions of (\ref{eq1}). Namely, in the same way as (\ref{eq4}), we can design
\begin{equation*}\label{}
\aligned
\bm{U}_{k+1}
&=\left(I_{q}-\widetilde{\bm{F}}\bm{G}\right)\bm{U}_{k}+\widetilde{\bm{F}}\widetilde{\bm{Y}}_{d}\\
&=\left(I_{q}-\widetilde{\bm{F}}\bm{G}\right)\bm{U}_{k}+\widetilde{\bm{F}}\bm{G}\bm{G}^{(1,3)}\bm{Y}_{d},\quad\forall k\in\mathbb{Z}_{+}
\endaligned
\end{equation*}

\noindent for some gain matrix $\widetilde{\bm{F}}\in\mathbb{R}^{q\times p}$. Clearly, it can be viewed as a special design of (\ref{eq4}) by setting $\bm{F}=\widetilde{\bm{F}}\bm{G}\bm{G}^{(1,3)}$ since $\bm{F}\bm{G}=\widetilde{\bm{F}}\bm{G}$ is also achieved by this selection of $\bm{F}$ due to $\bm{G}^{(1,3)}\in\bm{G}\{1,3\}$. Furthermore, the application of (\ref{eq4}) naturally avoids calculating a $\{1,3\}$-inverse $\bm{G}^{(1,3)}$ of $\bm{G}$.
\end{rem}

Although the iterative algorithm (\ref{eq4}) applies to any LAE (\ref{eq1}), its application to obtaining the least squares solutions requires a special selection of $\bm{F}$ in the form of $\bm{F}=\widetilde{\bm{F}}\bm{G}\bm{G}^{(1,3)}$  for some $\widetilde{\bm{F}}\in\mathbb{R}^{q\times p}$ and $\bm{G}^{(1,3)}\in\bm{G}\{1,3\}$, for which the following lemma gives an alternative selection condition.

\begin{lem}\label{lem11}
For the iterative algorithm (\ref{eq4}), there exists some matrix $\widetilde{\bm{F}}\in\mathbb{R}^{q\times p}$ such that $\bm{F}=\widetilde{\bm{F}}\bm{G}\bm{G}^{(1,3)}$, $\forall\bm{G}^{(1,3)}\in\bm{G}\{1,3\}$ if and only if the selection of $\bm{F}$ possesses the following property: 
\begin{enumerate}
\item[(P)] $\sn\bm{F}^{\tp}\subseteq\sn\bm{G}$. 
\end{enumerate}
\end{lem}

\begin{proof}
{\it Necessity:} With $\left(\bm{G}\bm{G}^{(1,3)}\right)^{\tp}=\bm{G}\bm{G}^{(1,3)}$, $\forall\bm{G}^{(1,3)}\in\bm{G}\{1,3\}$, $\bm{F}=\widetilde{\bm{F}}\bm{G}\bm{G}^{(1,3)}$ implies $\bm{F}^{\tp}=\bm{G}\left(\bm{G}^{(1,3)}\widetilde{\bm{F}}^{\tp}\right)$, and thus $\sn\bm{F}^{\tp}\subseteq\sn\bm{G}$ is direct.

{\it Sufficiency:} If $\sn\bm{F}^{\tp}\subseteq\sn\bm{G}$, then $\bm{F}=\overline{\bm{F}}\bm{G}^{\tp}$ holds for some $\overline{\bm{F}}\in\mathbb{R}^{q\times q}$. Thanks to $\bm{G}\bm{G}^{(1,3)}\bm{G}=\bm{G}$ and  $\left(\bm{G}\bm{G}^{(1,3)}\right)^{\tp}=\bm{G}\bm{G}^{(1,3)}$, $\forall\bm{G}^{(1,3)}\in\bm{G}\{1,3\}$, we can get $\bm{F}=\overline{\bm{F}}\left(\bm{G}\bm{G}^{(1,3)}\bm{G}\right)^{\tp}=\widetilde{\bm{F}}\bm{G}\bm{G}^{(1,3)}$ for $\widetilde{\bm{F}}=\overline{\bm{F}}\bm{G}^{\tp}$.
\end{proof}

By Lemma \ref{lem11}, we can employ the property (P) for the design of the iterative algorithm (\ref{eq4}), instead of adopting the particular design $\bm{F}=\widetilde{\bm{F}}\bm{G}\bm{G}^{(1,3)}$, which thus removes the dependence of (\ref{eq4}) on using a $\{1,3\}$-inverse $\bm{G}^{(1,3)}$ of $\bm{G}$ simultaneously.

\section{Observer-Based Iterative Solution Methods and Results of LAEs}\label{sec3}

In this section, we incorporate the ideas of observability and observability decomposition into the solving of any LAEs, and establish iterative solution methods and results simultaneously. We explore conditions for both exponentially fast convergence and monotonic convergence of our iterative algorithm.

\subsection{Basic Observability-Based Design and Analysis}

This subsection is devoted to the solving of the LAE (\ref{eq1}) for the observable case of the system (\ref{eq2}). For this case, Lemma \ref{lem7} ensures that $\bm{F}\bm{G}$ is nonsingular when the selection of $\bm{F}$ fulfills the condition (\ref{eq5}). By this observation and due to the equivalent conditions in Lemma \ref{lem7}, we can develop an observability-based iterative solution result for the LAE (\ref{eq1}) in the theorem below.

\begin{thm}\label{thm1}
For the LAE (\ref{eq1}), the sequence of $\bm{U}_{k}$, $\forall k\in\mathbb{Z}_{+}$ generated by (\ref{eq4}) converges exponentially fast with a converged value independent of the initial condition $\bm{U}_{0}$ as
\begin{equation}\label{eq6}
\bm{U}_{\infty}
\triangleq\lim_{k\to\infty}\bm{U}_{k}
=\left(\bm{F}\bm{G}\right)^{-1}\bm{F}\bm{Y}_{d}
\end{equation}

\noindent if and only if there exists some gain matrix $\bm{F}\in\mathbb{R}^{q\times p}$ fulfilling (\ref{eq5}). Furthermore, $\bm{U}_{\infty}\in\mathcal{O}_{O}$ always holds, and $\bm{U}_{\infty}$ is the unique solution for the LAE (\ref{eq1}) if and only if it is solvable. Otherwise, $\bm{U}_{\infty}$ is the unique least squares solution for any unsolvable LAE (\ref{eq1}) if and only if the property (P) holds.
%
\end{thm}

\begin{proof}
If $\bm{U}_{k}$ converges and (\ref{eq6}) holds, then the nonsingularity of $\bm{F}\bm{G}$ ensures that $\bm{G}$ has full-column rank. Hence, from Lemma \ref{lem7}, there always exists some gain matrix $\bm{F}\in\mathbb{R}^{q\times p}$ such that (\ref{eq5}) holds. On the contrary, if (\ref{eq5}) holds, then from (\ref{eq4}), we equivalently have
\begin{equation}\label{eq7}
\aligned
\bm{U}_{k}-\left(\bm{F}\bm{G}\right)^{-1}\bm{F}\bm{Y}_{d}
&=\left(I_{q}-\bm{F}\bm{G}\right)\left[\bm{U}_{k-1}-\left(\bm{F}\bm{G}\right)^{-1}\bm{F}\bm{Y}_{d}\right]\\
&=\left(I_{q}-\bm{F}\bm{G}\right)^{k}\Big[\bm{U}_{0}
-\left(\bm{F}\bm{G}\right)^{-1}\bm{F}\bm{Y}_{d}\Big],\quad\forall k\in\mathbb{Z}_{+}.
\endaligned
\end{equation}

\noindent By substituting (\ref{eq5}) into (\ref{eq7}), we can conclude that $\bm{U}_{k}$ converges exponentially fast for the system (\ref{eq4}), where its converged value takes the form of (\ref{eq6}). In addition, (\ref{eq6}) holds regardless of any selection of $\bm{U}_{0}$, i.e., $\bm{U}_{\infty}$ is independent of $\bm{U}_{0}$. By specifically taking $\bm{\chi}=\bm{G}\left(\bm{G}^{\tp}\bm{G}\right)^{-1}\left(\bm{F}\bm{G}\right)^{-1}\bm{F}\bm{Y}_{d}\in\mathbb{R}^{p}$, we can validate
\[\bm{G}^{\tp}\bm{\chi}
=\left(\bm{F}\bm{G}\right)^{-1}\bm{F}\bm{Y}_{d}
=\bm{U}_{\infty}\]

\noindent which obviously results in $\bm{U}_{\infty}\in\rm{span}\bm{G}^{\tp}$. Therefore, $\bm{U}_{\infty}\in\mathcal{O}_{O}$ is immediate by Lemma \ref{lem6}.

Since $\bm{G}$ has the full-column rank, the solution (respectively, least squares solution) for the LAE (\ref{eq1}) is unique if it is solvable (respectively, unsolvable). This unique solution or least squares solution takes the form of $\bm{U}_{d}=\left(\bm{G}^{\tp}\bm{G}\right)^{-1}\bm{G}^{\tp}\bm{Y}_{d}$ according to Lemma \ref{lem8}. When (\ref{eq1}) is a solvable LAE, the substitution of (\ref{eq1}) into (\ref{eq6}) leads to
\[
\bm{G}\bm{U}_{\infty}
=\bm{G}\left(\bm{F}\bm{G}\right)^{-1}\bm{F}\left(\bm{G}\bm{U}_{d}\right)
=\bm{G}\bm{U}_{d}
=\bm{Y}_{d}.
\]

\noindent Hence, $\bm{U}_{\infty}$ in (\ref{eq6}) is the unique solution for the LAE (\ref{eq1}) if and only if it is solvable. Otherwise, for any unsolvable LAE (\ref{eq1}), if the property (P) holds, then there exists some square matrix $\overline{\bm{F}}\in\mathbb{R}^{q\times q}$ that fulfills $\bm{F}=\overline{\bm{F}}\bm{G}^{\tp}$, and hence $\overline{\bm{F}}=\bm{F}\bm{G}\left(\bm{G}^{\tp}\bm{G}\right)^{-1}$ follows and is naturally nonsingular. As a consequence,
\[
\bm{U}_{\infty}
=\left(\overline{\bm{F}}\bm{G}^{\tp}\bm{G}\right)^{-1}\overline{\bm{F}}\bm{G}^{\tp}\bm{Y}_{d}
=\left(\bm{G}^{\tp}\bm{G}\right)^{-1}\bm{G}^{\tp}\bm{Y}_{d}
\]

\noindent is the least squares solution to (\ref{eq1}). On the contrary, if $\bm{U}_{\infty}$ in (\ref{eq6}) is the least squares solution to (\ref{eq1}), then it follows $\left(\bm{F}\bm{G}\right)^{-1}\bm{F}=\left(\bm{G}^{\tp}\bm{G}\right)^{-1}\bm{G}^{\tp}$, which leads to $\bm{F}=\overline{\bm{F}}\bm{G}^{\tp}$ for $\overline{\bm{F}}=\bm{F}\bm{G}\left(\bm{G}^{\tp}\bm{G}\right)^{-1}$. Consequently, $\sn\bm{F}^{\tp}\subseteq\sn\bm{G}$ is immediate. The statements of Theorem \ref{thm1} are all developed.
\end{proof}

\begin{rem}\label{rem1}
From Theorem \ref{thm1} together with Lemmas \ref{lem7} and \ref{lem8}, a perspective from the system observability is exploited for the design of iterative algorithms to solve LAEs with full-column rank mapping matrices despite the solvability of them. Thanks to the dual relation between observability and controllability, we can develop the synthesis of the iterative algorithm (\ref{eq4}) by employing the popular design methods for feedback controllers. In particular, the exponential convergence result of Theorem \ref{thm1} can be improved to provide a monotonic convergence result through the synthesis of $\bm{F}$, for which the condition (\ref{eq5}) needs to be strengthened to satisfy $\left\|I_{q}-\bm{F}\bm{G}\right\|<1$ for some compatible matrix norm $\|\cdot\|$. 
\end{rem}

For the (least squares) solution result of LAEs in Theorem \ref{thm1}, we give a candidate selection of $\bm{F}$ in the following corollary.

\begin{cor}\label{cor5}
For the iterative algorithm (\ref{eq4}), let $\bm{F}=\bm{\sigma}\bm{G}^{\tp}$ be selected for $\bm{\sigma}\in\left(0,2/\lambda_{\max}\left(\bm{G}^{\tp}\bm{G}\right)\right)$. Then the same results of Theorem \ref{thm1} can be established, where the iteration convergence of (\ref{eq4}) can be further achieved in a monotonic way, that is, for $\bm{\gamma}=\max_{i=1,2,\cdots,q}\left\{\left|1-\bm{\sigma}\lambda_{i}\left(\bm{G}^{\tp}\bm{G}\right)\right|\right\}<1$,
\begin{equation*}\label{}
\aligned
\left\|\bm{U}_{k+1}-\left(\bm{F}\bm{G}\right)^{-1}\bm{F}\bm{Y}_{d}\right\|_{2}
&\leq\bm{\gamma}\left\|\bm{U}_{k}-\left(\bm{F}\bm{G}\right)^{-1}\bm{F}\bm{Y}_{d}\right\|_{2},\quad\forall k\in\mathbb{Z}_{+}.
\endaligned
\end{equation*}
\end{cor}

%
\subsection{Design and Analysis With Observability Decomposition}

Despite the design results of iterative algorithms in Theorem \ref{thm1}, the LAE (\ref{eq1}) is solved based on the requirement of the full-column rank condition of $\bm{G}$ or the observability of the system (\ref{eq4}) and (\ref{eq3}). Clearly, when this requirement can not be satisfied, the design results of Theorem \ref{thm1} may not be applied any longer. Moreover, a question naturally arising is: whether and how can the iterative algorithm (\ref{eq4}) be explored for solving the LAE (\ref{eq1}) without any requirements on $\bm{G}$? The answer to this question is affirmative, where the application of the iterative algorithm (\ref{eq4}) is closely tied to the observability decomposition thanks to the development of Lemma \ref{lem1}. It actually coincides with a common fact that only observable states can be approached by the state observers. To clearly disclose these facts, we introduce a basic solvability result for the LAE (\ref{eq1}) in the presence of any $\bm{Y}_{d}\in\mathbb{R}^{p}$ in the following lemma.

\begin{lem}\label{lem2}
The LAE (\ref{eq1}) is solvable for any $\bm{Y}_{d}\in\mathbb{R}^{p}$ if and only if any of the following conditions holds.
\begin{enumerate}
%
\item
There exists some matrix $\bm{G}^{+}\in\mathbb{R}^{q\times p}$ to yield $I_{p}=\bm{G}\bm{G}^{+}$.

\item
The matrix $\bm{G}$ is of full-row rank, i.e., $\rank\left(\bm{G}\right)=p$.
%

\item
There exists some gain matrix $\bm{F}\in\mathbb{R}^{q\times p}$ such that
\begin{equation}\label{eq8}
\rho\left(I_{p}-\bm{G}\bm{F}\right)<1.
\end{equation}

\item
There exists some full-row rank matrix $\bm{M}\in\mathbb{R}^{(q-p)\times q}$ to construct a nonsingular matrix as
\[
\bm{P}^{-1}=\begin{bmatrix}\bm{G}\\\bm{M}\end{bmatrix}\in\mathbb{R}^{q\times q}
\]

\noindent such that the system (\ref{eq2}) has a standard observability decomposition (that is, standard form for the unobservable system \cite{am:06}) with respect to $\overline{\bm{U}}_{d}(k)=\bm{P}^{-1}\bm{U}_{d}(k)$ as
\begin{equation}\label{eq36}
\left\{\aligned
\overline{\bm{U}}_{d}(k+1)&=\overline{\bm{U}}_{d}(k)+\begin{bmatrix}\bm{G}\\0\end{bmatrix}\bm{V}_{d}(k)\\
\bm{Y}_{d}(k)&=\begin{bmatrix}I_{p}&0\end{bmatrix}\overline{\bm{U}}_{d}(k)
\endaligned\right.,\quad\forall k\in\mathbb{Z}_{+}.
\end{equation}
%
\end{enumerate}
\end{lem}

\begin{proof}
The equivalence relations among all the statements in this lemma can be easily developed with the standard matrix theory \cite{bg:03,lt:85} and linear system theory \cite{r:96,am:06}, and thus the detailed proofs are omitted for simplicity.
\end{proof}

From Lemma \ref{lem2}, it is clear that we may no longer apply the design results of Theorem \ref{thm1} to obtain the solutions for the LAE (\ref{eq1}) in the absence of the observability property because we can not gain any estimation information for the unobservable states by (\ref{eq36}). However, this drawback may not affect the application of the iterative algorithm (\ref{eq4}) to solving the LAE (\ref{eq1}), especially due to the specific structure of (\ref{eq36}). To make this observation clear to follow, we combine the result of Lemma \ref{lem1} to introduce an observability decomposition lemma for the system formed by (\ref{eq4}) and (\ref{eq3}).

\begin{lem}\label{lem9}
Let the same linear transformation as (\ref{eq36}) be peformed. Then the system formed by (\ref{eq4}) and (\ref{eq3}) has a standard observability decomposition with respect to $\overline{\bm{U}}_{k}=\bm{P}^{-1}\bm{U}_{k}$ as 
\begin{equation}\label{eq53}
\left\{\aligned
\overline{\bm{U}}_{k+1}
&=\begin{bmatrix}I_{p}-\bm{G}\bm{F}&0\\0&I_{q-p}\end{bmatrix}\overline{\bm{U}}_{k}
+\begin{bmatrix}\bm{G}\bm{F}\\0\end{bmatrix}\bm{Y}_{d}\\
\bm{Y}_{k}
&=\begin{bmatrix}I_{p}&0\end{bmatrix}\overline{\bm{U}}_{k}
\endaligned\right.,\quad\forall k\in\mathbb{Z}_{+}
\end{equation}

\noindent if and only if the condition (\ref{eq8}) holds.
\end{lem}

\begin{proof}
The necessity and sufficiency is obvious from the equivalence between the results 3) and 4) in Lemma \ref{lem2}. Hence, we only need to develop (\ref{eq53}). In fact, the definition of $\bm{P}^{-1}=\left[\bm{G}^{\tp},\bm{M}^{\tp}\right]^{\tp}$ implies
\[\overline{\bm{G}}=\bm{G}\bm{P}=\begin{bmatrix}I_{p}&0\end{bmatrix}
\]

\noindent which, together with (\ref{eq3}), immediately leads to
\begin{equation}\label{eq37}
\bm{Y}_{k}=\begin{bmatrix}I_{p}&0\end{bmatrix}\overline{\bm{U}}_{k},\quad\forall k\in\mathbb{Z}_{+}.
\end{equation}

\noindent Because (\ref{eq8}) ensures that $\bm{G}\bm{F}$ is nonsingular, there exists some full-column rank matrix $\bm{N}\in\mathbb{R}^{q\times(q-p)}$ such that we can choose $\bm{P}$ in a specific form of
\[
\bm{P}=\begin{bmatrix}\bm{F}\left(\bm{G}\bm{F}\right)^{-1}&\bm{N}\end{bmatrix}\in\mathbb{R}^{q\times q}.
\]

\noindent Obviously, we can validate $\bm{M}\bm{F}=0$, and consequently,
\[
\overline{\bm{F}}=\bm{P}^{-1}\bm{F}=\begin{bmatrix}\bm{G}\bm{F}\\0\end{bmatrix}.
\]

\noindent We thus revisit Lemma \ref{lem1} and can obtain that for any $\bm{Y}_{d}\in\mathbb{R}^{p}$,
\begin{equation}\label{eq11}
\aligned
\overline{\bm{U}}_{k+1}
&=\left(I_{q}-\overline{\bm{F}}\,\overline{\bm{G}}\right)\overline{\bm{U}}_{k}+\overline{\bm{F}}\bm{Y}_{d}\\
&=\begin{bmatrix}I_{p}-\bm{G}\bm{F}&0\\0&I_{q-p}\end{bmatrix}\overline{\bm{U}}_{k}
+\begin{bmatrix}\bm{G}\bm{F}\\0\end{bmatrix}\bm{Y}_{d},\quad\forall k\in\mathbb{Z}_{+}.
\endaligned
\end{equation}

\noindent From (\ref{eq37}) and (\ref{eq11}), (\ref{eq53}) is immediate.
\end{proof}

From the observability decomposition (\ref{eq53}) in Lemma \ref{lem9}, we can clearly verify that the iterative algorithm (\ref{eq4}) is essentially effective in approaching the $p$ observable state variables of the system \eqref{eq2}. It simultaneously follows from (\ref{eq53}) that the $(q-p)$ unobservable state variables essentially keep unchanged along the iteration axis. With these observations, we can present the following theorem to develop a solution result to the LAE (\ref{eq1}) without the observability hypothesis of its counterpart system (\ref{eq2}), which is realized with a different design condition of the iterative algorithm (\ref{eq4}) from that provided in Theorem \ref{thm1}.

\begin{thm}\label{thm2}
For any $\bm{Y}_{d}\in\mathbb{R}^{p}$, the LAE (\ref{eq1}) is solvable such that the sequence of $\bm{U}_{k}$, $\forall k\in\mathbb{Z}_{+}$ generated by (\ref{eq4}) converges exponentially fast with its converged value depending heavily on the initial condition $\bm{U}_{0}$ and forming a set given by
\begin{equation}\label{eq9}
\aligned
\mathcal{U}_{\mathrm{ILC}}(\bm{Y}_{d})
=\bigg\{\bm{U}_{\infty}&=\left[I_{q}-\bm{F}\left(\bm{G}\bm{F}\right)^{-1}\bm{G}\right]\bm{U}_{0}
+\bm{F}\left(\bm{G}\bm{F}\right)^{-1}\bm{Y}_{d}
\Big|\bm{U}_{0}\in\mathbb{R}^{q}\bigg\},\quad\forall\bm{Y}_{d}\in\mathbb{R}^{p}
\endaligned
\end{equation}

\noindent if and only if there exists some gain matrix $\bm{F}\in\mathbb{R}^{q\times p}$ fulfilling (\ref{eq8}). Furthermore, $\mathcal{U}_{\mathrm{ILC}}(\bm{Y}_{d})=\mathcal{U}_{d}(\bm{Y}_{d})$, $\forall\bm{Y}_{d}\in\mathbb{R}^{p}$ holds, where $\mathcal{U}_{d}\left(\bm{Y}_{d}\right)$ is the set of all solutions to the solvable LAE (\ref{eq1}), i.e.,
\begin{equation}\label{eq10}
\mathcal{U}_{d}\left(\bm{Y}_{d}\right)
=\left\{\bm{U}_{d}\in\mathbb{R}^{q}
\big|\bm{Y}_{d}=\bm{G}\bm{U}_{d}\right\},\quad\forall\bm{Y}_{d}\in\mathbb{R}^{p}.
\end{equation}
\end{thm}

\begin{proof}
From Lemma \ref{lem2}, we know the equivalence between the spectral radius condition (\ref{eq8}) and the solvability of the LAE (\ref{eq1}). By revisiting $\bm{P}^{-1}=\left[\bm{G}^{\tp},\bm{M}^{\tp}\right]^{\tp}$ and $\bm{P}=\left[\bm{F}\left(\bm{G}\bm{F}\right)^{-1},\bm{N}\right]$, we can also get $\bm{M}\bm{N}=I_{q-p}$, $\bm{G}\bm{N}=0$, and $\bm{F}\left(\bm{G}\bm{F}\right)^{-1}\bm{G}+\bm{N}\bm{M}=I_{q}$. If we denote $\overline{\bm{U}}_{2,k}=\bm{M}\bm{U}_{k}$, then we have
\[
\overline{\bm{U}}_{k}
=\begin{bmatrix}\bm{G}\bm{U}_{k}\\\bm{M}\bm{U}_{k}\end{bmatrix}
=\begin{bmatrix}\bm{Y}_{k}\\\overline{\bm{U}}_{2,k}\end{bmatrix}
\]

\noindent and consequently, the use of (\ref{eq53}) results in two subsystems as
\begin{equation}\label{eq12}
\bm{Y}_{k+1}
=\left(I_{p}-\bm{G}\bm{F}\right)\bm{Y}_{k}+\bm{G}\bm{F}\bm{Y}_{d},\quad\forall k\in\mathbb{Z}_{+}
\end{equation}

\noindent and
\begin{equation}\label{eq13}
\overline{\bm{U}}_{2,k+1}
=\overline{\bm{U}}_{2,k},\quad\forall k\in\mathbb{Z}_{+}.
\end{equation}

\noindent By integrating (\ref{eq8}) into (\ref{eq12}), we can conclude that $\bm{Y}_{k}$ converges exponentially fast with
\begin{equation}\label{eq14}
\bm{Y}_{\infty}
\triangleq\lim_{k\to\infty}\bm{Y}_{k}
=\bm{Y}_{d}
\end{equation}

\noindent and as a direct consequence of (\ref{eq13}), we have
\begin{equation}\label{eq15}
\overline{\bm{U}}_{2,k}
=\overline{\bm{U}}_{2,0}
=\bm{M}\bm{U}_{0},\quad\forall k\in\mathbb{Z}_{+}.
\end{equation}

\noindent Based on (\ref{eq14}) and (\ref{eq15}), we can leverage $\bm{U}_{k}=\bm{P}\overline{\bm{U}}_{k}$ and $\overline{\bm{U}}_{k}=\left[\bm{Y}_{k}^{\tp},\overline{\bm{U}}_{2,k}^{\tp}\right]^{\tp}$ to conclude that $\bm{U}_{k}$ converges exponentially fast with $\bm{U}_{\infty}$ given by
\begin{equation}\label{eq16}
\aligned
\bm{U}_{\infty}
&=\bm{P}\begin{bmatrix}
\lim\limits_{k\to\infty}\bm{Y}_{k}\\\lim\limits_{k\to\infty}\overline{\bm{U}}_{2,k}\end{bmatrix}\\
&=\bm{F}\left(\bm{G}\bm{F}\right)^{-1}\bm{Y}_{\infty}
+\bm{N}\overline{\bm{U}}_{2,0}\\
&=\bm{F}\left(\bm{G}\bm{F}\right)^{-1}\bm{Y}_{d}
+\bm{N}\bm{M}\bm{U}_{0}\\
&=\left[I_{q}-\bm{F}\left(\bm{G}\bm{F}\right)^{-1}\bm{G}\right]\bm{U}_{0}
+\bm{F}\left(\bm{G}\bm{F}\right)^{-1}\bm{Y}_{d}.
\endaligned
\end{equation}

\noindent From (\ref{eq9}) and (\ref{eq16}), we have $\bm{U}_{\infty}\in\mathcal{U}_{\mathrm{ILC}}(\bm{Y}_{d})$ for any $\bm{Y}_{d}\in\mathbb{R}^{p}$.

Besides, we can employ (\ref{eq16}) to get $\bm{G}\bm{U}_{\infty}=\bm{Y}_{d}$. This ensures $\bm{U}_{\infty}\in\mathcal{U}_{d}(\bm{Y}_{d})$, which implies $\mathcal{U}_{\mathrm{ILC}}(\bm{Y}_{d})\subseteq\mathcal{U}_{d}(\bm{Y}_{d})$. In addition, for any $\bm{U}_{d}\in\mathcal{U}_{d}(\bm{Y}_{d})$, we notice (\ref{eq10}) and can deduce
\begin{equation}\label{eq17}
\aligned
\bm{U}_{d}
&=\left[I_{q}-\bm{F}\left(\bm{G}\bm{F}\right)^{-1}\bm{G}\right]\bm{U}_{d}
+\bm{F}\left(\bm{G}\bm{F}\right)^{-1}\bm{G}\bm{U}_{d}\\
&=\left[I_{q}-\bm{F}\left(\bm{G}\bm{F}\right)^{-1}\bm{G}\right]\bm{U}_{0}
+\bm{F}\left(\bm{G}\bm{F}\right)^{-1}\bm{Y}_{d}
\endaligned
\end{equation}

\noindent where $\bm{U}_{0}=\bm{U}_{d}+\bm{F}\bm{\alpha}$, $\forall\bm{\alpha}\in\mathbb{R}^{p}$. Thus,  $\bm{U}_{d}\in\mathcal{U}_{\mathrm{ILC}}(\bm{Y}_{d})$ follows from (\ref{eq9}) and (\ref{eq17}), from which we have $\mathcal{U}_{d}(\bm{Y}_{d})\subseteq\mathcal{U}_{\mathrm{ILC}}(\bm{Y}_{d})$. As a consequence, $\mathcal{U}_{d}(\bm{Y}_{d})=\mathcal{U}_{\mathrm{ILC}}(\bm{Y}_{d})$ holds.
\end{proof}

\begin{rem}\label{rem2}
By Theorem \ref{thm2}, we clearly disclose that although the system given by (\ref{eq4}) and (\ref{eq3}) is not observable, it is possible to solve the LAE (\ref{eq1}) still based on the synthesis of this system if the LAE \eqref{eq1} is solvable for any given $\bm{Y}_d\in\mathbb{R}^p$. Furthermore, the implementation of the iterative algorithm (\ref{eq4}) can arrive at all solutions to LAEs through the selection of different initial conditions, for which an analytical description is provided in (\ref{eq9}). It is worth highlighting that the calculation of the iterative algorithm (\ref{eq4}) is linear, and we only need to determine the gain matrix $\bm{F}$ under the selection condition (\ref{eq8}) through the popular feedback controller design tools. As a candidate selection of $\bm{F}$, $\bm{F}=\bm{\sigma}\bm{G}^{\tp}$ can be directly adopted for $\bm{\sigma}\in\left(0,2/\lambda_{\max}\left(\bm{G}\bm{G}^{\tp}\right)\right)$, which together with (\ref{eq9}) also yields a monotonic convergence result of the iterative algorithm (\ref{eq4}), namely,
\begin{equation*}\label{}
\aligned
\left\|\bm{U}_{k+1}-\bm{U}_{\infty}\right\|_{2}
&=\left\|\left(I_{q}-\bm{F}\bm{G}\right)\left(\bm{U}_{k}-\bm{U}_{\infty}\right)\right\|_{2}\\
&=\left\|\left(I_{q}-\bm{\sigma}\bm{G}^{\tp}\bm{G}\right)\left(\bm{U}_{k}-\bm{U}_{\infty}\right)\right\|_{2}\\
&\leq\left\|\bm{U}_{k}-\bm{U}_{\infty}\right\|_{2},\quad\forall k\in\mathbb{Z}_{+}.
\endaligned
\end{equation*}
\end{rem}
\begin{rem}\label{rem3}
By (\ref{eq9}), we denote that for any given $\bm{Y}_{d}\in\mathbb{R}^{p}$,
\[\aligned
\bm{U}_{\infty}^{\prime}
&=\bm{F}\left(\bm{G}\bm{F}\right)^{-1}\bm{Y}_{d}\\
\bm{U}_{\infty}^{\prime\prime}\left(\bm{U}_{0}\right)
&=\left[I_{q}-\bm{F}\left(\bm{G}\bm{F}\right)^{-1}\bm{G}\right]\bm{U}_{0},\quad\forall\bm{U}_{0}\in\mathbb{R}^{q}.
\endaligned\]

\noindent Since we can take $\bm{G}^{+}=\bm{F}\left(\bm{G}\bm{F}\right)^{-1}$ such that $\bm{G}\bm{G}^{+}=I_{p}$ holds, $\bm{U}_{\infty}^{\prime}$ is clearly a particular solution for the LAE (\ref{eq1}) despite any given $\bm{Y}_{d}\in\mathbb{R}^{p}$, which is determined essentially by the iteration convergence of $\bm{Y}_{k}$, as shown in (\ref{eq16}). From (\ref{eq53}), we know that $\bm{Y}_{k}$ collects the observable state variables. These observations reveal actually that the observable states of the system (\ref{eq2}) are tied tightly to the particular solutions for the LAE (\ref{eq1}). Besides, by contrast to $\bm{U}_{\infty}^{\prime}$, $\bm{U}_{\infty}^{\prime\prime}\left(\bm{U}_{0}\right)$ is a general solution for the LAE (\ref{eq1}) (i.e., a solution for the homogeneous equation $\bm{G}\bm{U}_{d}=0$). In view of (\ref{eq16}), $\bm{U}_{\infty}^{\prime\prime}\left(\bm{U}_{0}\right)$ is calculated essentially through the iteration convergence of $\overline{\bm{U}}_{2,k}$, where $\overline{\bm{U}}_{2,k}$ is actually composed of the unobservable state variables of the system (\ref{eq2}) by noting the equivalence relation between the unobservability subspace of the system \eqref{eq2} and that of the system \eqref{eq4} and \eqref{eq3}. Moreover, a fact worth highlighting is that $\overline{\bm{U}}_{2,k}$ remains unchanged for all iterations, as shown in (\ref{eq15}), and as a result, $\bm{U}_{\infty}^{\prime\prime}\left(\bm{U}_{0}\right)$ depends linearly on the initial condition $\bm{U}_{0}$.
\end{rem}

As a consequence of Lemma \ref{lem6} and Theorem \ref{thm2}, the following corollary presents a deterministic relation between iterative solutions of LAEs and the unobservability subspace.

\begin{cor}\label{cor6}
For the LAE (\ref{eq1}), let the iterative algorithm (\ref{eq4}) be applied under the condition (\ref{eq8}). Then it follows
\[
\mathcal{O}_{NO}
=\sn\left[I_{q}-\bm{F}\left(\bm{G}\bm{F}\right)^{-1}\bm{G}\right]
\]

\noindent which implies
\[
\mathcal{U}_{d}(\bm{Y}_{d})
=\mathcal{U}_{\mathrm{ILC}}(\bm{Y}_{d})
=\mathcal{O}_{NO}+\bm{F}\left(\bm{G}\bm{F}\right)^{-1}\bm{Y}_{d}, \quad\forall\bm{Y}_{d}\in\mathbb{R}^{p}.
\]
\end{cor}

\begin{proof}
Thanks to the following relation:
\[
\nl\bm{G}=\sn\bm{N}=\sn\left(\bm{N}\bm{M}\right)=
\sn\left[I_{q}-\bm{F}\left(\bm{G}\bm{F}\right)^{-1}\bm{G}\right]\]

\noindent this corollary is immediate from Lemma \ref{lem6} and Theorem \ref{thm2}.
\end{proof}

From Corollary \ref{cor6}, we can derive that for any given $\bm{Y}_{d}\in\mathbb{R}^{p}$, the solution set of the LAE (\ref{eq1}) is parallel to the unobservability subspace of the system (\ref{eq2}) with an offset $\bm{F}(\bm{G}\bm{F})^{-1}\bm{Y}_{d}$. Of note is that this offset is exactly a particular solution to the LAE (\ref{eq1}), and has a close connection to the observable state of the system (\ref{eq2}) (see Remark \ref{rem3}). Furthermore, we have $\bm{F}\left(\bm{G}\bm{F}\right)^{-1}\bm{Y}_{d}\in\mathcal{O}_{O}$, $\forall\bm{Y}_{d}\in\mathbb{R}^{p}$ if and only if $\sn\bm{F}=\sn\bm{G}^{\tp}$, that is, an observable state of the system (\ref{eq2}) denotes exactly the distance between the solution set of the LAE (\ref{eq1}) and the unobservability subspace of the system (\ref{eq2}) once the selection of $\bm{F}$ fulfills $\sn\bm{F}=\sn\bm{G}^{\tp}$.

In Theorems \ref{thm1} and \ref{thm2}, the solving problem of the LAE (\ref{eq1}) is considered based on the iterative algorithm (\ref{eq4}) for the cases of $\bm{G}$ with full-column rank and full-row rank, respectively. Next, we aim at solving the LAE (\ref{eq1}) with the iterative algorithm (\ref{eq4}) for general cases without requiring any full rank conditions of $\bm{G}$. Before proceeding further with this discussion, we provide a helpful lemma about the rank decomposition of $\bm{G}$ \cite{lt:85}.

\begin{lem}\label{lem3}
Let $\rank\left(\bm{G}\right)=m$. Then two matrices $\bm{H}\in\mathbb{R}^{p\times m}$ and $\widehat{\bm{G}}\in\mathbb{R}^{m\times q}$ can be determined such that
\begin{equation}\label{eq18}
\bm{G}=\bm{H}\widehat{\bm{G}}
\end{equation}

\noindent where $\rank\left(\bm{H}\right)=\rank\left(\widehat{\bm{G}}\right)=m$. Moreover, there can always be designed some gain matrix $\bm{F}\in\mathbb{R}^{q\times p}$ such that
\begin{equation}\label{eq39}
\rho\left(I_{m}-\widehat{\bm{G}}\bm{F}\bm{H}\right)<1.
\end{equation}
%
\end{lem}

\begin{proof}
We can obtain (\ref{eq18}) as a direct consequence of the use of the full rank decomposition technique of matrices (see, e.g., Proposition 3 of \cite[Subsection 3.8]{lt:85}). Due to $\rank\left(\bm{H}\right)=\rank\left(\widehat{\bm{G}}\right)=m$, $\bm{H}$ and $\widehat{\bm{G}}$ in (\ref{eq18}) are full-column and full-row rank matrices, respectively. Hence, we can accomplish (\ref{eq39}) by particularly taking $\bm{F}$ in the form of
\begin{equation}\label{eq38}
\bm{F}=\widehat{\bm{G}}^{\tp}\left(\widehat{\bm{G}}\widehat{\bm{G}}^{\tp}\right)^{-1}\left(I_{m}-\widetilde{\bm{F}}\right)
\left(\bm{H}^{\tp}\bm{H}\right)^{-1}\bm{H}^{\tp}
\end{equation}

\noindent for any matrix $\widetilde{\bm{F}}\in\mathbb{R}^{m\times m}$ such that $\rho\left(\widetilde{\bm{F}}\right)<1$. This obviously implies that (\ref{eq39}) always holds for some $\bm{F}\in\mathbb{R}^{q\times p}$. 
\end{proof}

\begin{rem}\label{rem4}
Similarly to the derivation of (\ref{eq38}), if we resort to the full rank properties of $\bm{H}$ and $\widehat{\bm{G}}$, then we can benefit from (\ref{eq18}) to deduce $\widehat{\bm{G}}=\left(\bm{H}^{\tp}\bm{H}\right)^{-1}\bm{H}^{\tp}\bm{G}$ and  $\bm{H}=\bm{G}\widehat{\bm{G}}^{\tp}\left(\widehat{\bm{G}}\widehat{\bm{G}}^{\tp}\right)^{-1}$. It is thus immediate to arrive at
\begin{equation}\label{eq54}
\bm{G}=\bm{H}\left(\bm{H}^{\tp}\bm{H}\right)^{-1}\bm{H}^{\tp}\bm{G}
=\bm{G}\widehat{\bm{G}}^{\tp}\left(\widehat{\bm{G}}\widehat{\bm{G}}^{\tp}\right)^{-1}\widehat{\bm{G}}
\end{equation}

\noindent where both $\bm{H}\left(\bm{H}^{\tp}\bm{H}\right)^{-1}\bm{H}^{\tp}$ and $\widehat{\bm{G}}^{\tp}\left(\widehat{\bm{G}}\widehat{\bm{G}}^{\tp}\right)^{-1}\widehat{\bm{G}}$ are idempotent matrices. Besides, we have $\sn\bm{G}=\sn\bm{H}$ and $\sn\bm{G}^{\tp}=\sn\widehat{\bm{G}}^{\tp}$ from the full rank properties of $\widehat{\bm{G}}$ and $\bm{H}$, respectively.
\end{rem}

To proceed further with Lemma \ref{lem3}, we can introduce a least squares solution result of any unsolvable LAE (\ref{eq1}).

\begin{lem}\label{lem10}
For any $\bm{Y}_{d}\in\mathbb{R}^{p}$, the least squares solutions of the unsolvable LAE (\ref{eq1}) can be described in the form of
\begin{equation}\label{eq29}
\aligned
\bm{U}_{d}
&=\arg{\min_{\bm{\Omega}\in\mathbb{R}^{q}}}\left\|\bm{Y}_{d}-\bm{G}\bm{\Omega}\right\|_{2}\\
&=\widehat{\bm{G}}^{\tp}\left(\widehat{\bm{G}}\widehat{\bm{G}}^{\tp}\right)^{-1}\left(\bm{H}^{\tp}\bm{H}\right)^{-1}\bm{H}^{\tp}\bm{Y}_{d}
+\bm{\gamma},\quad\forall\bm{\gamma}\in\rm{null}\bm{G}.
\endaligned
\end{equation}

\noindent If (\ref{eq39}) is fulfilled, then (\ref{eq29}) can be described alternatively as
\begin{equation}\label{eq56}
\aligned
\bm{U}_{d}
&=\arg{\min_{\bm{\Omega}\in\mathbb{R}^{q}}}\left\|\bm{Y}_{d}-\bm{G}\bm{\Omega}\right\|_{2}\\
&=\bm{F}\bm{H}\left(\bm{H}^{\tp}\bm{G}\bm{F}\bm{H}\right)^{-1}\bm{H}^{\tp}\bm{Y}_{d}
+\bm{\gamma},\quad\forall\bm{\gamma}\in\rm{null}\bm{G}.
\endaligned
\end{equation}
\end{lem}

\begin{proof}
For any $\bm{\Omega}\in\mathbb{R}^{q}$, we have $\widehat{\bm{\Omega}}=\widehat{\bm{G}}\bm{\Omega}\in\mathbb{R}^{m}$, by which $\min_{\bm{\Omega}\in\mathbb{R}^{q}}\left\|\bm{Y}_{d}-\bm{H}\widehat{\bm{G}}\bm{\Omega}\right\|_{2}
\geq\min_{\widehat{\bm{\Omega}}\in\mathbb{R}^{m}}\left\|\bm{Y}_{d}-\bm{H}\widehat{\bm{\Omega}}\right\|_{2}$ holds. With the full-row rank property of $\widehat{\bm{G}}$, we can verify that for any $\widehat{\bm{\Omega}}\in\mathbb{R}^{m}$, $\bm{\Omega}=\widehat{\bm{G}}^{\tp}\left(\widehat{\bm{G}}\widehat{\bm{G}}^{\tp}\right)^{-1}\widehat{\bm{\Omega}}\in\mathbb{R}^{q}$ fulfills $\widehat{\bm{G}}\bm{\Omega}=\widehat{\bm{\Omega}}$, and therefore $\min_{\bm{\Omega}\in\mathbb{R}^{q}}\left\|\bm{Y}_{d}-\bm{H}\widehat{\bm{G}}\bm{\Omega}\right\|_{2}
\leq\min_{\widehat{\bm{\Omega}}\in\mathbb{R}^{m}}\left\|\bm{Y}_{d}-\bm{H}\widehat{\bm{\Omega}}\right\|_{2}$ holds. From (\ref{eq18}), we can consequently arrive at
\begin{equation}\label{eq25}
\aligned
\min_{\bm{\Omega}\in\mathbb{R}^{q}}\left\|\bm{Y}_{d}-\bm{G}\bm{\Omega}\right\|_{2}
&=\min_{\bm{\Omega}\in\mathbb{R}^{q}}\left\|\bm{Y}_{d}-\bm{H}\widehat{\bm{G}}\bm{\Omega}\right\|_{2}\\
&=\min_{\widehat{\bm{\Omega}}\in\mathbb{R}^{m}}\left\|\bm{Y}_{d}-\bm{H}\widehat{\bm{\Omega}}\right\|_{2}.
\endaligned
\end{equation}

\noindent Since $\bm{H}$ is a full-column rank matrix, the least squares solution to the unsolvable LAE $\bm{Y}_{d}=\bm{H}\widehat{\bm{U}}_{d}$ is unique, which is actually given by (for the same reason as the proof of Lemma \ref{lem8})
\begin{equation}\label{eq26}
\widehat{\bm{U}}_{d}
=\arg{\min_{\widehat{\bm{\Omega}}\in\mathbb{R}^{m}}}\left\|\bm{Y}_{d}-\bm{H}\widehat{\bm{\Omega}}\right\|_{2}
=\left(\bm{H}^{\tp}\bm{H}\right)^{-1}\bm{H}^{\tp}\bm{Y}_{d}.
\end{equation}

\noindent By incorporating (\ref{eq26}) into (\ref{eq25}), we can thus develop that $\bm{U}_{d}=\arg{\min_{\bm{\Omega}\in\mathbb{R}^{q}}}\left\|\bm{Y}_{d}-\bm{G}\bm{\Omega}\right\|_{2}$ is the solution to the following LAE:
\begin{equation}\label{eq57}
\widehat{\bm{G}}\bm{U}_{d}=\left(\bm{H}^{\tp}\bm{H}\right)^{-1}\bm{H}^{\tp}\bm{Y}_{d}.
\end{equation}

\noindent Again by the full-row rank property of $\widehat{\bm{G}}$, the LAE (\ref{eq57}) always has solutions that can be expressed in the form of (see Lemma \ref{lem2} and Theorem 2 of \cite[Subsection 3.10]{lt:85})
\[
\bm{U}_{d}=\widehat{\bm{G}}^{\tp}\left(\widehat{\bm{G}}\widehat{\bm{G}}^{\tp}\right)^{-1}\left(\bm{H}^{\tp}\bm{H}\right)^{-1}\bm{H}^{\tp}\bm{Y}_{d}
+\bm{\gamma},\quad\forall\bm{\gamma}\in\rm{null}\widehat{\bm{G}}
\]

\noindent which, together with $\bm{\gamma}\in\rm{null}\widehat{\bm{G}}\Leftrightarrow\bm{\gamma}\in\rm{null}\bm{G}$ resulting from the full-column rank property of $\bm{H}$, yields (\ref{eq29}). To proceed, if (\ref{eq39}) holds, $\widehat{\bm{G}}\bm{F}\bm{H}$ is nonsingular, and thus, $\bm{F}\bm{H}\left(\bm{H}^{\tp}\bm{G}\bm{F}\bm{H}\right)^{-1}\bm{H}^{\tp}\bm{Y}_{d}$ is a particular solution of the LAE (\ref{eq57}) thanks to
\[\aligned
\widehat{\bm{G}}\bm{F}\bm{H}\left(\bm{H}^{\tp}\bm{G}\bm{F}\bm{H}\right)^{-1}\bm{H}^{\tp}\bm{Y}_{d}
&=\widehat{\bm{G}}\bm{F}\bm{H}\left(\bm{H}^{\tp}\bm{H}\widehat{\bm{G}}\bm{F}\bm{H}\right)^{-1}\bm{H}^{\tp}\bm{Y}_{d}\\
&=\left(\bm{H}^{\tp}\bm{H}\right)^{-1}\bm{H}^{\tp}\bm{Y}_{d}.
\endaligned\]

\noindent Consequently, we can get (\ref{eq56}) for the same reason as (\ref{eq29}).
\end{proof}

By Theorem \ref{thm2} and with Lemmas \ref{lem3} and \ref{lem10}, we show a general solution result of the LAE (\ref{eq1}) in the case of any rank condition $\rank\left(\bm{G}\right)=m$ ($m\neq0$ and $m\leq\min\{p,q\}$) and any $\bm{Y}_{d}\in\mathbb{R}^{p}$.

\begin{thm}\label{thm3}
For the LAE (\ref{eq1}) with any given $\bm{Y}_{d}\in\mathbb{R}^{p}$, let the iterative algorithm (\ref{eq4}) that satisfies the property (P) be applied. Then the sequence of $\bm{U}_{k}$, $\forall k\in\mathbb{Z}_{+}$ generated by (\ref{eq4}) converges exponentially fast with its converged value depending heavily on the initial condition $\bm{U}_{0}$ and forming a set given by
\begin{equation}\label{eq20}
\aligned
\mathcal{U}_{\mathrm{ILC}}(\bm{Y}_{d})
=\bigg\{\bm{U}_{\infty}&=\left[I_{q}-\bm{F}\bm{H}\left(\bm{H}^{\tp}\bm{G}\bm{F}\bm{H}\right)^{-1}\bm{H}^{\tp}\bm{G}\right]\bm{U}_{0}\\
&~~~+\bm{F}\bm{H}\left(\bm{H}^{\tp}\bm{G}\bm{F}\bm{H}\right)^{-1}\bm{H}^{\tp}\bm{Y}_{d}
\Big|\bm{U}_{0}\in\mathbb{R}^{q}\bigg\},\quad\forall\bm{Y}_{d}\in\mathbb{R}^{p}
\endaligned
\end{equation}

\noindent if and only if the spectral radius condition (\ref{eq39}) holds. Besides, the resulting sequence of $\bm{Y}_{k}$, $\forall k\in\mathbb{Z}_{+}$ converges to
\begin{equation}\label{eq21}
\bm{Y}_{\infty}
=\bm{H}\left(\bm{H}^{\tp}\bm{H}\right)^{-1}\bm{H}^{\tp}\bm{Y}_{d},\quad\forall\bm{Y}_{d}\in\mathbb{R}^{p}.
\end{equation}

\noindent Furthermore, $\mathcal{U}_{\mathrm{ILC}}(\bm{Y}_{d})=\mathcal{U}_{d}(\bm{Y}_{d})$ holds if and only if the LAE (\ref{eq1}) is solvable; and $\mathcal{U}_{\mathrm{ILC}}\left(\bm{Y}_{d}\right)=\overline{\mathcal{U}}_{d}\left(\bm{Y}_{d}\right)$ follows, otherwise, where $\overline{\mathcal{U}}_{d}\left(\bm{Y}_{d}\right)$ denotes the set of all the least squares solutions to the unsolvable LAE (\ref{eq1}), namely,
\begin{equation}\label{eq27}
\aligned
\overline{\mathcal{U}}_{d}\left(\bm{Y}_{d}\right)
=\bigg\{\bm{U}_{d}\in\mathbb{R}^{q}
\Big|\left\|\bm{Y}_{d}-\bm{G}\bm{U}_{d}\right\|_{2}
&=\min_{\bm{\Omega}\in\mathbb{R}^{q}}\left\|\bm{Y}_{d}-\bm{G}\bm{\Omega}\right\|_{2}\\
&>0\bigg\},\quad\forall\bm{Y}_{d}\in\mathbb{R}^{p}.
\endaligned
\end{equation}
\end{thm}

\begin{proof}
Owing to the property (P), we can obtain $\bm{F}^{\tp}=\bm{G}\bm{Q}$ for some $\bm{Q}\in\mathbb{R}^{q\times q}$, which together with (\ref{eq54}) results in
\[
\aligned
\bm{F}^{\tp}
=\bm{G}\bm{Q}
=\bm{H}\left(\bm{H}^{\tp}\bm{H}\right)^{-1}\bm{H}^{\tp}\bm{G}\bm{Q}
=\bm{H}\left(\bm{H}^{\tp}\bm{H}\right)^{-1}\bm{H}^{\tp}\bm{F}^{\tp}
\endaligned\]

\noindent i.e., $\bm{F}=\bm{F}\bm{H}\left(\bm{H}^{\tp}\bm{H}\right)^{-1}\bm{H}^{\tp}$. Due to this fact and by substituting (\ref{eq18}) into (\ref{eq4}), we can derive
\begin{equation}\label{eq22}
\bm{U}_{k+1}
=\left(I_{q}-\bm{F}\bm{H}\widehat{\bm{G}}\right)\bm{U}_{k}
+\bm{F}\bm{H}\left(\bm{H}^{\tp}\bm{H}\right)^{-1}\bm{H}^{\tp}\bm{Y}_{d},\quad\forall k\in\mathbb{Z}_{+}.
\end{equation}

\noindent In the same way as the proof of Theorem \ref{thm2}, we know from (\ref{eq22}) that $\bm{U}_{k}$ converges exponentially fast if and only if the spectral radius condition (\ref{eq39}) holds. Simultaneously, $\bm{U}_{\infty}$ satisfies
\begin{equation}\label{eq23}
\aligned
\bm{U}_{\infty}
&=\left[I_{q}-\bm{F}\bm{H}\left(\widehat{\bm{G}}\bm{F}\bm{H}\right)^{-1}\widehat{\bm{G}}\right]\bm{U}_{0}\\
&~~~+\bm{F}\bm{H}\left(\widehat{\bm{G}}\bm{F}\bm{H}\right)^{-1}\left(\bm{H}^{\tp}\bm{H}\right)^{-1}\bm{H}^{\tp}\bm{Y}_{d},
\quad\forall\bm{Y}_{d}\in\mathbb{R}^{p}.
\endaligned
\end{equation}

\noindent Since $\widehat{\bm{G}}=\left(\bm{H}^{\tp}\bm{H}\right)^{-1}\bm{H}^{\tp}\bm{G}$ follows from (\ref{eq18}), the substitution of it into (\ref{eq23}) yields
\begin{equation*}\label{}
\aligned
\bm{U}_{\infty}
&=\left[I_{q}-\bm{F}\bm{H}\left(\bm{H}^{\tp}\bm{G}\bm{F}\bm{H}\right)^{-1}\bm{H}^{\tp}\bm{G}\right]\bm{U}_{0}\\
&~~~+\bm{F}\bm{H}\left(\bm{H}^{\tp}\bm{G}\bm{F}\bm{H}\right)^{-1}\bm{H}^{\tp}\bm{Y}_{d},\quad\forall\bm{Y}_{d}\in\mathbb{R}^{p}
\endaligned
\end{equation*}

\noindent with which we can obtain (\ref{eq20}). By resorting to (\ref{eq18}) and (\ref{eq23}), we can also arrive at
\begin{equation*}\label{}
\aligned
\bm{Y}_{\infty}
&=\left(\bm{H}\widehat{\bm{G}}\right)\bm{U}_{\infty}\\
&=\bm{H}\widehat{\bm{G}}\Bigg\{\left[I_{q}-\bm{F}\bm{H}\left(\widehat{\bm{G}}\bm{F}\bm{H}\right)^{-1}\widehat{\bm{G}}\right]\bm{U}_{0}\\
&~~~+\bm{F}\bm{H}\left(\widehat{\bm{G}}\bm{F}\bm{H}\right)^{-1}\left(\bm{H}^{\tp}\bm{H}\right)^{-1}\bm{H}^{\tp}\bm{Y}_{d}\Bigg\}\\
&=\bm{H}\left(\bm{H}^{\tp}\bm{H}\right)^{-1}\bm{H}^{\tp}\bm{Y}_{d},\quad\forall\bm{Y}_{d}\in\mathbb{R}^{p}.
\endaligned
\end{equation*}

\noindent Namely, (\ref{eq21}) holds.

For any given $\bm{Y}_{d}\in\mathbb{R}^{p}$, if the LAE (\ref{eq1}) is solvable, then the use of (\ref{eq21}) directly leads to
\[\bm{G}\bm{U}_{\infty}
=\bm{Y}_{\infty}
=\bm{H}\left(\bm{H}^{\tp}\bm{H}\right)^{-1}\bm{H}^{\tp}\left(\bm{H}\widehat{\bm{G}}\bm{U}_{d}\right)
=\bm{H}\widehat{\bm{G}}\bm{U}_{d}
=\bm{Y}_{d}
\]

\noindent which implies $\mathcal{U}_{\mathrm{ILC}}(\bm{Y}_{d})\subseteq\mathcal{U}_{d}(\bm{Y}_{d})$. In addition, for any $\bm{U}_{d}\in\mathcal{U}_{d}(\bm{Y}_{d})$, i.e., $\bm{Y}_{d}=\bm{G}\bm{U}_{d}$, we can validate
\begin{equation}\label{eq24}
\aligned
\bm{U}_{d}
&=\left[I_{q}-\bm{F}\bm{H}\left(\bm{H}^{\tp}\bm{G}\bm{F}\bm{H}\right)^{-1}\bm{H}^{\tp}\bm{G}\right]\bm{U}_{d}\\
&~~~+\bm{F}\bm{H}\left(\bm{H}^{\tp}\bm{G}\bm{F}\bm{H}\right)^{-1}\bm{H}^{\tp}\bm{G}\bm{U}_{d}\\
&=\left[I_{q}-\bm{F}\bm{H}\left(\bm{H}^{\tp}\bm{G}\bm{F}\bm{H}\right)^{-1}\bm{H}^{\tp}\bm{G}\right]\bm{U}_{d}\\
&~~~+\bm{F}\bm{H}\left(\bm{H}^{\tp}\bm{G}\bm{F}\bm{H}\right)^{-1}\bm{H}^{\tp}\bm{Y}_{d}\\
&=\left[I_{q}-\bm{F}\bm{H}\left(\bm{H}^{\tp}\bm{G}\bm{F}\bm{H}\right)^{-1}\bm{H}^{\tp}\bm{G}\right]\bm{U}_{0}\\
&~~~+\bm{F}\bm{H}\left(\bm{H}^{\tp}\bm{G}\bm{F}\bm{H}\right)^{-1}\bm{H}^{\tp}\bm{Y}_{d}\\
\endaligned
\end{equation}

\noindent where $\bm{U}_{0}=\bm{U}_{d}+\bm{F}\bm{H}\bm{\beta}$, $\forall\bm{\beta}\in\mathbb{R}^{m}$. Thus, (\ref{eq20}) and (\ref{eq24}) ensure $\bm{U}_{d}\in\mathcal{U}_{\mathrm{ILC}}(\bm{Y}_{d})$, from which $\mathcal{U}_{d}(\bm{Y}_{d})\subseteq\mathcal{U}_{\mathrm{ILC}}(\bm{Y}_{d})$ holds. With $\mathcal{U}_{\mathrm{ILC}}(\bm{Y}_{d})\subseteq\mathcal{U}_{d}(\bm{Y}_{d})$ and $\mathcal{U}_{d}(\bm{Y}_{d})\subseteq\mathcal{U}_{\mathrm{ILC}}(\bm{Y}_{d})$, it is immediate to get $\mathcal{U}_{d}(\bm{Y}_{d})=\mathcal{U}_{\mathrm{ILC}}(\bm{Y}_{d})$. Conversely, if $\mathcal{U}_{\mathrm{ILC}}(\bm{Y}_{d})=\mathcal{U}_{d}(\bm{Y}_{d})$ holds, then $\mathcal{U}_{d}(\bm{Y}_{d})\neq{\O}$ follows directly because $\mathcal{U}_{\mathrm{ILC}}(\bm{Y}_{d})\neq\O$ is obvious thanks to (\ref{eq20}).

For any given $\bm{Y}_{d}\in\mathbb{R}^{p}$, if the LAE (\ref{eq1}) is unsolvable, then by incorporating the result of (\ref{eq26}), we can combine (\ref{eq20}) and (\ref{eq21}) with (\ref{eq25}) to arrive at
\begin{equation}\label{eq28}
\aligned
\left\|\bm{Y}_{d}-\bm{G}\bm{U}_{\infty}\right\|_{2}
&=\left\|\bm{Y}_{d}-\bm{Y}_{\infty}\right\|_{2}\\
&=\left\|\bm{Y}_{d}-\bm{H}\left(\bm{H}^{\tp}\bm{H}\right)^{-1}\bm{H}^{\tp}\bm{Y}_{d}\right\|_{2}\\
&=\left\|\bm{Y}_{d}-\bm{H}\widehat{\bm{U}}_{d}\right\|_{2}\\
&=\min_{\widehat{\bm{\Omega}}\in\mathbb{R}^{m}}\left\|\bm{Y}_{d}-\bm{H}\widehat{\bm{\Omega}}\right\|_{2}\\
&=\min_{\bm{\Omega}\in\mathbb{R}^{q}}\left\|\bm{Y}_{d}-\bm{G}\bm{\Omega}\right\|_{2}
,\quad\forall\bm{U}_{\infty}\in\mathcal{U}_{\mathrm{ILC}}(\bm{Y}_{d}).
\endaligned
\end{equation}

\noindent It follows from (\ref{eq28}) that $\bm{U}_{\infty}\in\overline{\mathcal{U}}_{d}(\bm{Y}_{d})$, and thus $\mathcal{U}_{\mathrm{ILC}}(\bm{Y}_{d})\subseteq\overline{\mathcal{U}}_{d}(\bm{Y}_{d})$ holds. In addition, for any $\bm{U}_{d}\in\overline{\mathcal{U}}_{d}(\bm{Y}_{d})$, we notice (\ref{eq56}), and can follow the same lines as (\ref{eq24}) to derive
\begin{equation}\label{eq30}
\aligned
\bm{U}_{d}
&=\left[I_{q}-\bm{F}\bm{H}\left(\bm{H}^{\tp}\bm{G}\bm{F}\bm{H}\right)^{-1}\bm{H}^{\tp}\bm{G}\right]\bm{U}_{d}\\
&~~~+\bm{F}\bm{H}\left(\bm{H}^{\tp}\bm{G}\bm{F}\bm{H}\right)^{-1}\bm{H}^{\tp}\bm{G}\bm{U}_{d}\\
&=\left[I_{q}-\bm{F}\bm{H}\left(\bm{H}^{\tp}\bm{G}\bm{F}\bm{H}\right)^{-1}\bm{H}^{\tp}\bm{G}\right]\bm{U}_{d}\\
&~~~+\bm{F}\bm{H}\left(\bm{H}^{\tp}\bm{G}\bm{F}\bm{H}\right)^{-1}\bm{H}^{\tp}\bm{G}\bigg[\bm{\gamma}\\
&~~~+\bm{F}\bm{H}\left(\bm{H}^{\tp}\bm{G}\bm{F}\bm{H}\right)^{-1}\bm{H}^{\tp}\bm{Y}_{d}\bigg]\\
&=\left[I_{q}-\bm{F}\bm{H}\left(\bm{H}^{\tp}\bm{G}\bm{F}\bm{H}\right)^{-1}\bm{H}^{\tp}\bm{G}\right]\bm{U}_{0}\\
&~~~+\bm{F}\bm{H}\left(\bm{H}^{\tp}\bm{G}\bm{F}\bm{H}\right)^{-1}\bm{H}^{\tp}\bm{Y}_{d}\\
\endaligned
\end{equation}

\noindent where $\bm{U}_{0}=\bm{U}_{d}+\bm{F}\bm{H}\bm{\beta}$, $\forall\bm{\beta}\in\mathbb{R}^{m}$. Owing to (\ref{eq20}) and (\ref{eq30}), we can obtain $\bm{U}_{d}\in\mathcal{U}_{\mathrm{ILC}}(\bm{Y}_{d})$, which actually ensures $\overline{\mathcal{U}}_{d}(\bm{Y}_{d})\subseteq\mathcal{U}_{\mathrm{ILC}}(\bm{Y}_{d})$. This, together with $\mathcal{U}_{\mathrm{ILC}}(\bm{Y}_{d})\subseteq\overline{\mathcal{U}}_{d}(\bm{Y}_{d})$, leads to $\mathcal{U}_{\mathrm{ILC}}(\bm{Y}_{d})=\overline{\mathcal{U}}_{d}(\bm{Y}_{d})$ if the LAE (\ref{eq1}) is unsolvable. Conversely, if $\mathcal{U}_{\mathrm{ILC}}\left(\bm{Y}_{d}\right)=\overline{\mathcal{U}}_{d}\left(\bm{Y}_{d}\right)$, then $\mathcal{U}_{d}(\bm{Y}_{d})={\O}$ holds by (\ref{eq27}), that is, the LAE (\ref{eq1}) is unsolvable.
%
%
\end{proof}

Since there may exist different forms of the rank decomposition (\ref{eq18}), we proceed with Theorem \ref{thm3} to show that this does not affect the developed results of our iterative algorithm (\ref{eq4}).

\begin{cor}\label{cor7}
The result of Theorem \ref{thm3} is independent of the different selections of the rank decomposition in (\ref{eq18}). In other words, for any other rank decomposition of $\bm{G}$ given by
\begin{equation}\label{eq55}
\bm{G}=\overline{\bm{H}}\overline{\widehat{\bm{G}}}
\quad\hbox{with}~\rank\left(\overline{\bm{H}}\right)=\rank\left(\overline{\widehat{\bm{G}}}\right)
=m
\end{equation}

\noindent the same result of Theorem \ref{thm3} can still be obtained even though (\ref{eq55}) instead of (\ref{eq18}) is used. In particular, a candidate selection of $\bm{F}$ to satisfy the condition (\ref{eq39}) and the property (P) is given by $\bm{F}=\bm{\sigma}\bm{G}^{\tp}$ for some $\bm{\sigma}\in\left(0,2/\lambda_{\max}\left(\bm{G}^{\tp}\bm{G}\right)\right)$, which ensures the monotonic convergence of the iterative algorithm (\ref{eq4}), i.e.,
\begin{equation*}\label{}
\left\|\bm{U}_{k+1}-\bm{U}_{\infty}\right\|_{2}
\leq\left\|\bm{U}_{k}-\bm{U}_{\infty}\right\|_{2},\quad\forall k\in\mathbb{Z}_{+}.
\end{equation*}
\end{cor}

\begin{proof}
For any two rank decompositions (\ref{eq18}) and (\ref{eq55}) of $\bm{G}$, there exists some nonsingular matrix $\bm{\Phi}\in\mathbb{R}^{m\times m}$ such that
\[
\overline{\bm{H}}=\bm{H}\bm{\Phi},\quad
\overline{\widehat{\bm{G}}}=\bm{\Phi}^{-1}\widehat{\bm{G}}.
\]

\noindent This fact clearly renders the result of Theorem \ref{thm3} independent of the different selections of the rank decomposition of $\bm{G}$.

If $\bm{F}=\bm{\sigma}\bm{G}^{\tp}$ holds, then the property (P) is immediate. Since $\widehat{\bm{G}}\bm{G}^{\tp}\bm{H}$ has the same nonzero eigenvalues with $\bm{G}^{\tp}\bm{H}\widehat{\bm{G}}=\bm{G}^{\tp}\bm{G}$, it is obvious that all the eigenvalues of  $\widehat{\bm{G}}\bm{G}^{\tp}\bm{H}$ are nonnegative, and thus positive due to $\widehat{\bm{G}}\bm{G}^{\tp}\bm{H}= \widehat{\bm{G}}\widehat{\bm{G}}^{\tp}\bm{H}^{\tp}\bm{H}$. This ensures that (\ref{eq39}) holds if
\[
\left|1-\bm{\sigma}\lambda_{i}\left(\widehat{\bm{G}}\bm{G}^{\tp}\bm{H}\right)\right|<1,\quad\forall i=1,2,\cdots,m
\]

\noindent which can be guaranteed for $\bm{\sigma}\in\left(0,2/\lambda_{\max}\left(\bm{G}^{\tp}\bm{G}\right)\right)$ thanks to $\lambda_{\max}\left(\widehat{\bm{G}}\bm{G}^{\tp}\bm{H}\right)=\lambda_{\max}\left(\bm{G}^{\tp}\bm{G}\right)$. Further, this selection of $\bm{\sigma}$ yields $\left\|I_{q}-\bm{\sigma}\bm{G}^{\tp}\bm{G}\right\|_{2}\leq1$, where the equality holds if and only if $m=q$. We hence use (\ref{eq20}) and (\ref{eq22}), and can also deduce
\begin{equation*}
\aligned
\left\|\bm{U}_{k+1}-\bm{U}_{\infty}\right\|_{2}
&=\left\|\left(I_{q}-\bm{F}\bm{H}\widehat{\bm{G}}\right)\left(\bm{U}_{k}-\bm{U}_{\infty}\right)\right\|_{2}\\
&=\left\|\left(I_{q}-\bm{\sigma}\bm{G}^{\tp}\bm{G}\right)\left(\bm{U}_{k}-\bm{U}_{\infty}\right)\right\|_{2}\\
&\leq\left\|\bm{U}_{k}-\bm{U}_{\infty}\right\|_{2},\quad\forall k\in\mathbb{Z}_{+}.
\endaligned
\end{equation*}

\noindent That is, the monotonic convergence is achieved for the iterative algorithm (\ref{eq4}).
%
\end{proof}

\begin{rem}\label{rem5}
With Theorem \ref{thm3}, we present a systematic design method of the iterative algorithm (\ref{eq4}) for solving the LAE (\ref{eq1}), regardless of any rank conditions of $\bm{G}$ and any given $\bm{Y}_{d}\in\mathbb{R}^{p}$. This is thanks to reasonably incorporating the design of $\bm{F}$ and the rank decomposition of $\bm{G}$ into the application of Theorem \ref{thm2}. In contrast to the selection condition (\ref{eq8}) that can be used for Theorem \ref{thm2} to derive $\bm{F}$ in the state feedback design framework, (\ref{eq39}) is more like a design condition in the framework of output feedback, despite which the design of $\bm{F}$ is always feasible due to the full rank conditions of the rank decomposition (see also Lemma \ref{lem3} and Corollary \ref{cor7}). In addition, it is worth pointing out that the result of Theorem \ref{thm3} benefits from the linear calculation of the iterative algorithm (\ref{eq4}), where it actually does not require the property (P) in the presence of any solvable LAEs. Because of the satisfaction $\bm{Y}_{d}=\bm{G}\bm{U}_{d}$, we can use (\ref{eq4}) and (\ref{eq18}) to derive
\begin{equation*}\label{}
\bm{U}_{k+1}
=\left(I_{q}-\bm{F}\bm{H}\widehat{\bm{G}}\right)\bm{U}_{k}
+\bm{F}\bm{H}\left(\widehat{\bm{G}}\bm{U}_{d}\right),\quad\forall k\in\mathbb{Z}_{+}
\end{equation*}

\noindent and in the same way as the proof of Theorem \ref{thm3}, we can obtain its corresponding results developed for solvable LAEs.
\end{rem}
\begin{rem}\label{rem7}
Note that the observability decomposition works for developing Theorem \ref{thm3}. Because there exists some full-row rank matrix $\widehat{\bm{M}}\in\mathbb{R}^{(q-m)\times q}$ such that we can gain a nonsingular matrix $\widehat{\bm{P}}^{-1}=\left[\widehat{\bm{G}}^{\tp},\widehat{\bm{M}}^{\tp}\right]^{\tp}$, and then leverage (\ref{eq39}) to express its inverse matrix in a particular form of
$\widehat{\bm{P}}=\left[\bm{F}\bm{H}\left(\widehat{\bm{G}}\bm{F}\bm{H}\right)^{-1},\widehat{\bm{N}}\right]$ for some full-column rank matrix $\widehat{\bm{N}}\in\mathbb{R}^{q\times(q-m)}$, we consider $\widehat{\bm{U}}_{k}=\widehat{\bm{P}}^{-1}\bm{U}_{k}$, $\forall k\in\mathbb{Z}_{+}$ for (\ref{eq22}) and can follow the same lines as the proof of Lemma \ref{lem9} to arrive at
\begin{equation}\label{eq32}
\left\{\aligned
\widehat{\bm{U}}_{k+1}
&=\begin{bmatrix}I_{m}-\widehat{\bm{G}}\bm{F}\bm{H}&0\\0&I_{q-m}\end{bmatrix}\widehat{\bm{U}}_{k}
+\begin{bmatrix}\widehat{\bm{G}}\bm{F}\bm{H}\\0\end{bmatrix}\left(\bm{H}^{\tp}\bm{H}\right)^{-1}\bm{H}^{\tp}\bm{Y}_{d}\\
\widehat{\bm{Y}}_{k}
&=\begin{bmatrix}I_{m}&0\end{bmatrix}\widehat{\bm{U}}_{k}
\endaligned\right.
\end{equation}

\noindent where $\widehat{\bm{Y}}_{k}=\left(\bm{H}^{\tp}\bm{H}\right)^{-1}\bm{H}^{\tp}\bm{Y}_{k}$. With (\ref{eq32}), we can make the same discussions as Remark \ref{rem3} for the result of Theorem \ref{thm3} from the viewpoint of observability decomposition, as shown in Fig. \ref{fig2}. Specifically, for $\bm{U}_{\infty}$ in (\ref{eq20}), $\bm{F}\bm{H}\left(\bm{H}^{\tp}\bm{G}\bm{F}\bm{H}\right)^{-1}\bm{H}^{\tp}\bm{Y}_{d}$ leads to a particular (least squares) solution for any LAE (\ref{eq1}) (see also (\ref{eq56})), while the other term constitutes the solution space of the corresponding homogeneous equation $\bm{G}\bm{U}_{d}=0$ for (\ref{eq1}), giving
\[\aligned
\mathcal{O}_{NO}
&=\bigg\{\left[I_{q}-\bm{F}\bm{H}\left(\bm{H}^{\tp}\bm{G}\bm{F}\bm{H}\right)^{-1}\bm{H}^{\tp}\bm{G}\right]\bm{U}_{0}
\Big|\bm{U}_{0}\in\mathbb{R}^{q}\bigg\}\\
&=\sn\left[I_{q}-\bm{F}\bm{H}\left(\bm{H}^{\tp}\bm{G}\bm{F}\bm{H}\right)^{-1}\bm{H}^{\tp}\bm{G}\right].
\endaligned
\]
%
\end{rem}

\begin{figure}
\centering
\includegraphics[width=3.3in]{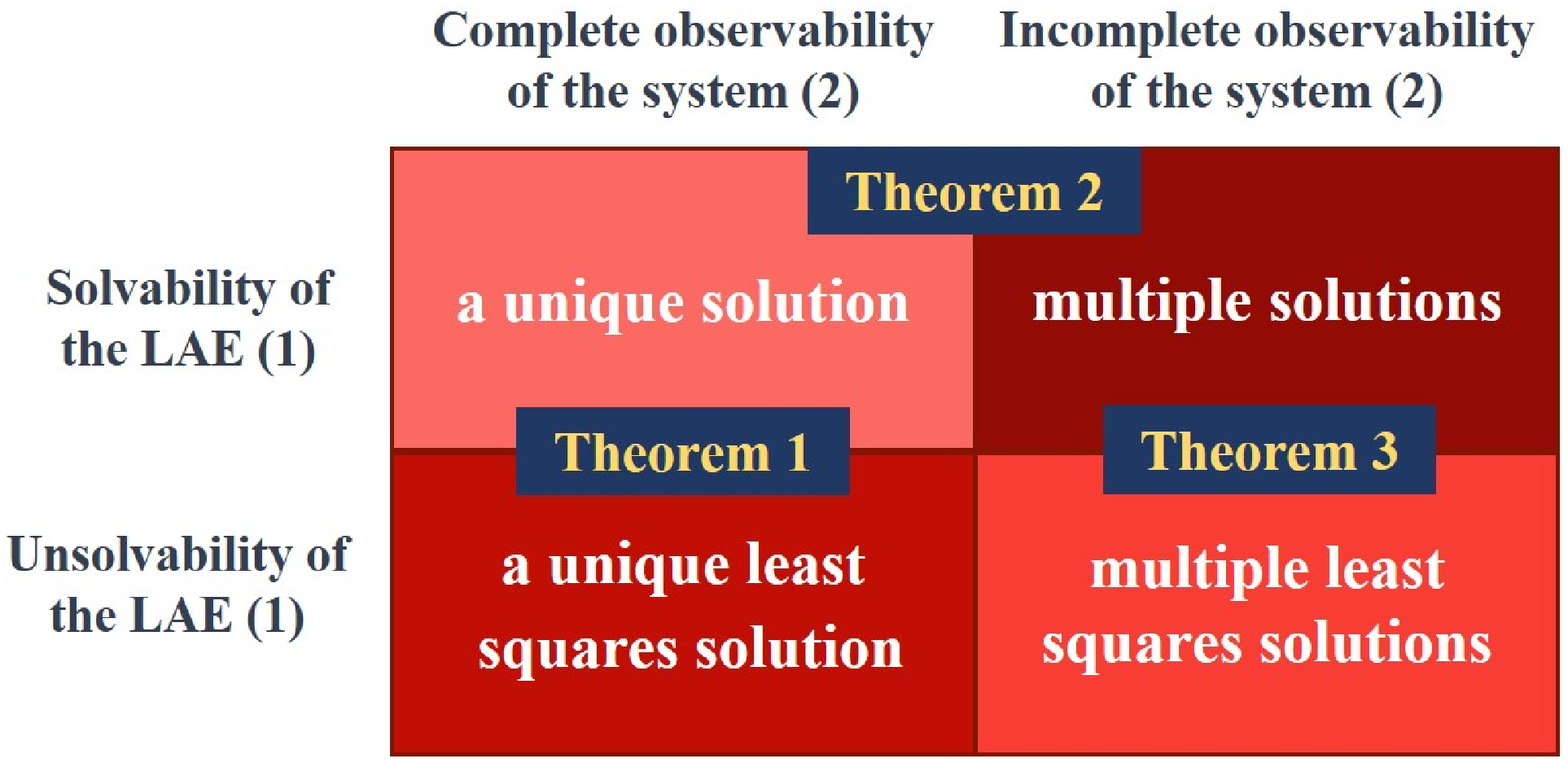}
\caption{An illustration of the relationship between the (least squares) solutions of the LAE (\ref{eq1}) and the observability properties of the system (\ref{eq2}).}\label{fig3}
\end{figure}

\begin{rem}\label{rem8}
It is worth emphasizing that based on Theorems \ref{thm1}-\ref{thm3}, how to carry out the iterative algorithm (\ref{eq4}) for solving any LAEs is developed by incorporating the ideas of both observer design and observability decomposition of the control systems. We clearly disclose the connection between the (least squares) solution to the LAE \eqref{eq1} and the observability of the system \eqref{eq2}, as shown by Fig. \ref{fig3}. To be specific, we incorporate the control design such that we can identify the solutions for any solvable LAEs and the least squares solutions for any unsolvable LAEs in a unified way, where we can determine all the (least squares) solutions for any LAEs through the different selections of the initial condition. Furthermore, we can regulate the convergence performance of our iterative algorithm by the proper selection of the gain matrix. These observations indicate the advantages of the iterative solution methods developed for LAEs under the idea of ``control design,'' which may provide new insights for enhancing the interaction between control and mathematics.
\end{rem}

\subsection{Implementation Algorithm for Solving LAEs}

With Theorems \ref{thm1}-\ref{thm3}, we introduce a specific implementation of the iterative algorithm \eqref{eq4} for calculating the (least squares) solutions to the LAE \eqref{eq1}, despite the solvability of it. We focus on the general case for the LAE (\ref{eq1}), where $\rank\left(\bm{G}\right)=m$ holds for any $m$ such that $m\neq0$ and $m\leq\min\{p,q\}$.

{
\renewcommand\arraystretch{1}{\center
\begin{tabular}{p{16cm}}
\specialrule{0.1em}{10pt}{5pt}
{\bf Algorithm 1: Solving the LAE \eqref{eq1} for any $\bm{Y}_{d}\in\mathbb{R}^{p}$}\\
\specialrule{0.1em}{5pt}{3pt}
\textbf{Input:} The mapping matrix $\bm{G}$, the vector $\bm{Y}_d$, the gain matrix $\bm{F}$ satisfying the property (P) and the condition (\ref{eq39}), the initial condition $\bm{U}_0$, and a tolerance $\varepsilon>0$.\\
\specialrule{0em}{0pt}{1.2pt}
\textbf{Output:} The (least squares) solution $\bm{U}_{\infty}$.\\
\specialrule{0em}{0pt}{1.2pt}
\quad $k\leftarrow0$\\
\specialrule{0em}{0pt}{1pt}
\quad\textbf{while} $k=0$ \textbf{or} $\left\|\bm{U}_{k}-\bm{U}_{k-1}\right\|_2\geq\varepsilon$ \textbf{do}\\
\specialrule{0em}{0pt}{1pt}
\quad\quad $\bm{U}_{k+1}\leftarrow\left(I_{q}-\bm{F}\bm{G}\right)\bm{U}_{k}+\bm{F}\bm{Y}_{d}$\\
\specialrule{0em}{0pt}{1pt}
\quad\quad$k\leftarrow k+1$\\
\specialrule{0em}{0pt}{1pt}
\quad\textbf{end while}\\
\quad$\bm{U}_{\infty}\leftarrow\bm{U}_{k}$\\
\specialrule{0.1em}{3pt}{10pt}
\end{tabular}}
}

\begin{rem}\label{rem11}
The Algorithm 1 can effectively work and apply to the solving of any LAEs, despite whether they are solvable or not. Actually, it applies the norm condition $\left\|\bm{U}_{k}-\bm{U}_{k-1}\right\|_{2}<\varepsilon$ as the criterion to stop its iteration, meaning that the variation between $\bm{U}_{k}$ and $\bm{U}_{k-1}$ decreases to be sufficiently small. This generally implies that $\bm{U}_{k}$ approximates to some (least squares) solution of the LAE (\ref{eq1}) in an acceptable tolerance. In addition, we may simultaneously take a norm condition $\left\|\bm{G}\bm{U}_k-\bm{Y}_d\right\|_{2}<\tilde{\varepsilon}$ into account for some tolerance $\tilde{\varepsilon}>0$, which can be adopted to help examine the solvability of the LAE (\ref{eq1}).
\end{rem}

To leverage the Algorithm 1 for solving LAEs, the key is to select a proper gain matrix $\bm{F}$, which fulfills both the property (P) and the spectral radius condition \eqref{eq39}. Based on Corollary \ref{cor7}, a candidate selection of $\bm{F}$ is given by $\bm{F}=\bm{\sigma}\bm{G}^{\tp}$ for some $\bm{\sigma}\in\left(0,2/\lambda_{\max}\left(\bm{G}^{\tp}\bm{G}\right)\right)$, which further leads to the monotonic convergence of the iterative algorithm \eqref{eq4}. But, for $\bm{G}$ with high dimensions, the derivation of $\lambda_{\max}\left(\bm{G}^{\tp}\bm{G}\right)$ may be challenging. To show a more direct criterion on how to properly choose $\bm{\sigma}$, we present the following lemma.

\begin{lem}\label{lem12}
If ${\bm{F}}=\bm{\sigma}{\bm{G}}^{\tp}$ is chosen for $\bm{\sigma}$ satisfying
\begin{equation*}\label{e7}
0<\bm{\sigma}<2/{\rm tr}\left({\bm{G}}{\bm{G}}^{\tp}\right)
\end{equation*}

\noindent then the spectral radius condition \eqref{eq39} and the property (P) can be guaranteed, and furthermore the monotonic convergence of the iterative algorithm \eqref{eq4} can be realized.
\end{lem}

\begin{proof}
Thanks to $\lambda_{\max}\left(\bm{G}^{\tp}\bm{G}\right)=\left\|\bm{G}\right\|_{2}^2\leq\left\|\bm{G}\right\|_{F}^2$, this lemma can be directly derived from Corollary \ref{cor7} by the equivalence between $\left\|\bm{G}\right\|_{F}^2$ and ${\rm tr}\left({\bm{G}}{\bm{G}}^{\tp}\right)$.
\end{proof}

Based on Lemma \ref{lem12}, we can gain ${\bm{F}}$ for the implementation of the Algorithm 1 through the simple calculation of ${\rm tr}\left(\bm{G}\bm{G}^{\tp}\right)$. By directly adopting this candidate selection of $\bm{F}$, we even do not need to perform the full rank decomposition of $\bm{G}$ in (\ref{eq18}) any longer, which may make the Algorithm 1 more available in practical applications. To clearly illustrate how the Algorithm 1 is applied, we present the following example for the solving of the LAE \eqref{eq1}.

{\it Example 1:} Consider the LAE \eqref{eq1} with
\[
\bm{G}=\left[\begin{matrix}
1&3&5&7&2\\
2&4&6&1&5\\
1&2&5&3&3\\
1&2&1&-2&2
\end{matrix}\right]
\]

\noindent for which it is easy to verify $\rank(\bm{G})=3$ and ${\rm tr}\left({\bm{G}}{\bm{G}}^{\tp}\right)=232$. Clearly, $\bm{G}$ is neither of full-row rank nor of full-column rank. To apply the Algorithm 1, we choose $\bm{F}=1/120\bm{G}^\tp$ according to Lemma \ref{lem12}, set the tolerance as $\varepsilon=10^{-5}$, and take the initial condition as $\bm{U}_0=[1,1,0,0,0]^{\tp}$. For the vector $\bm{Y}_{d}$, we consider two cases associated with the solvability and unsolvability of the LAE (\ref{eq1}), respectively.

{\it Case 1): $\bm{Y}_d=[1,0,2,-2]^{\tp}$.}

In this case, the operation of the Algorithm 1 ends after $587$ iterations because of $\left\|\bm{U}_{587}-\bm{U}_{586}\right\|_2<10^{-5}$. Simultaneously, we have
\[\bm{U}_{587}=\left[\frac{\Ds585}{\Ds2308},\frac{\Ds-496}{\Ds357},\frac{\Ds241}{\Ds248},\frac{\Ds103}{\Ds1846},\frac{\Ds-573}{\Ds3427}\right]^{\tp}
\]

\noindent which fulfills $\left\|\bm{Y}_d-\bm{G}\bm{U}_{587}\right\|_{2}=9.1018\times10^{-4}$. Consequently, $\bm{U}_{d}\approx\bm{U}_{587}$ can be regarded as a solution for the LAE \eqref{eq1}. This actually is consistent with the fact that the LAE (\ref{eq1}) is indeed solvable due to $\rank\left(\left[\bm{G}~\bm{Y}_d\right]\right)=\rank\left(\bm{G}\right)=3$.

\begin{figure}
\centering
\includegraphics[width=2.8in]{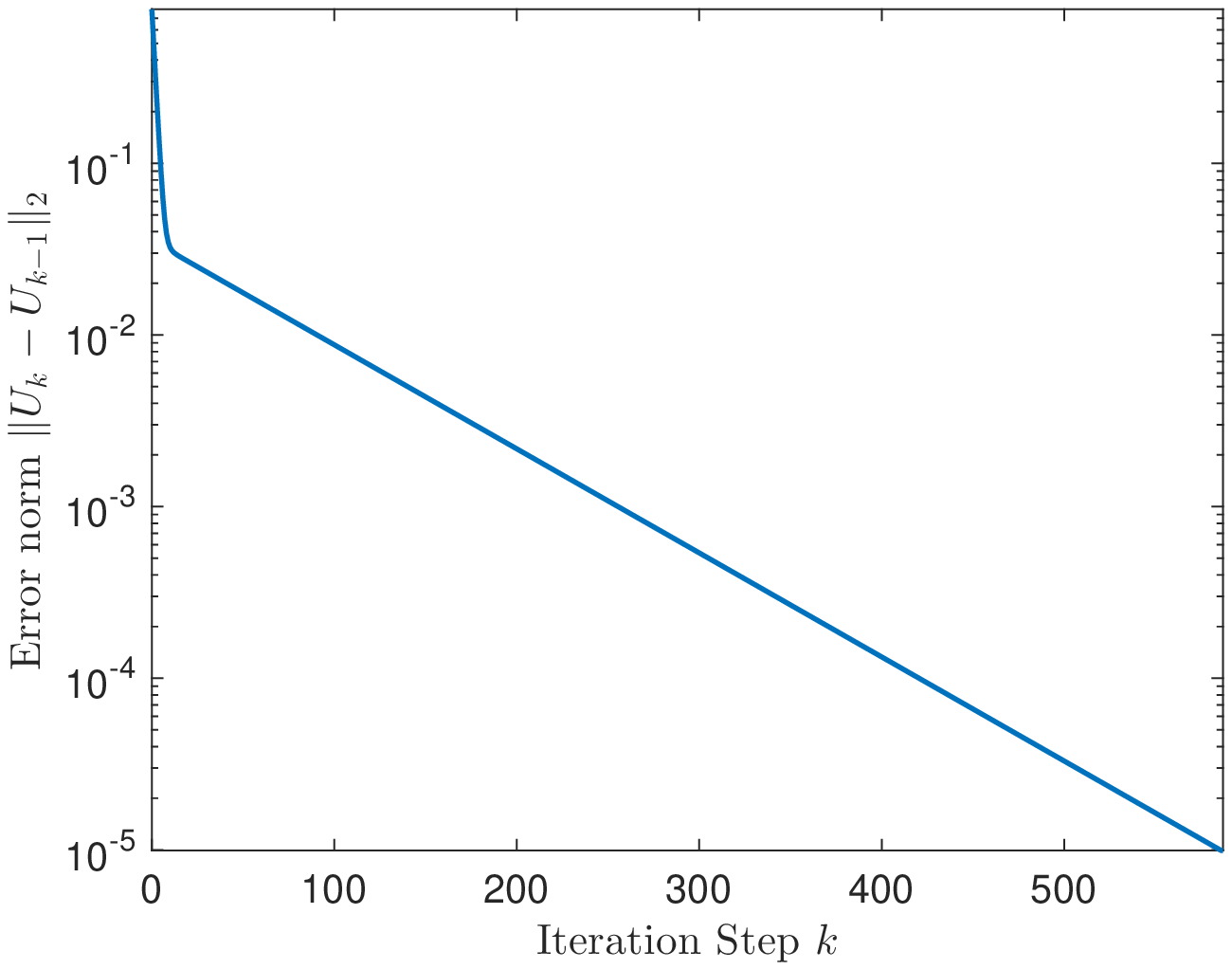}~\includegraphics[width=2.8in]{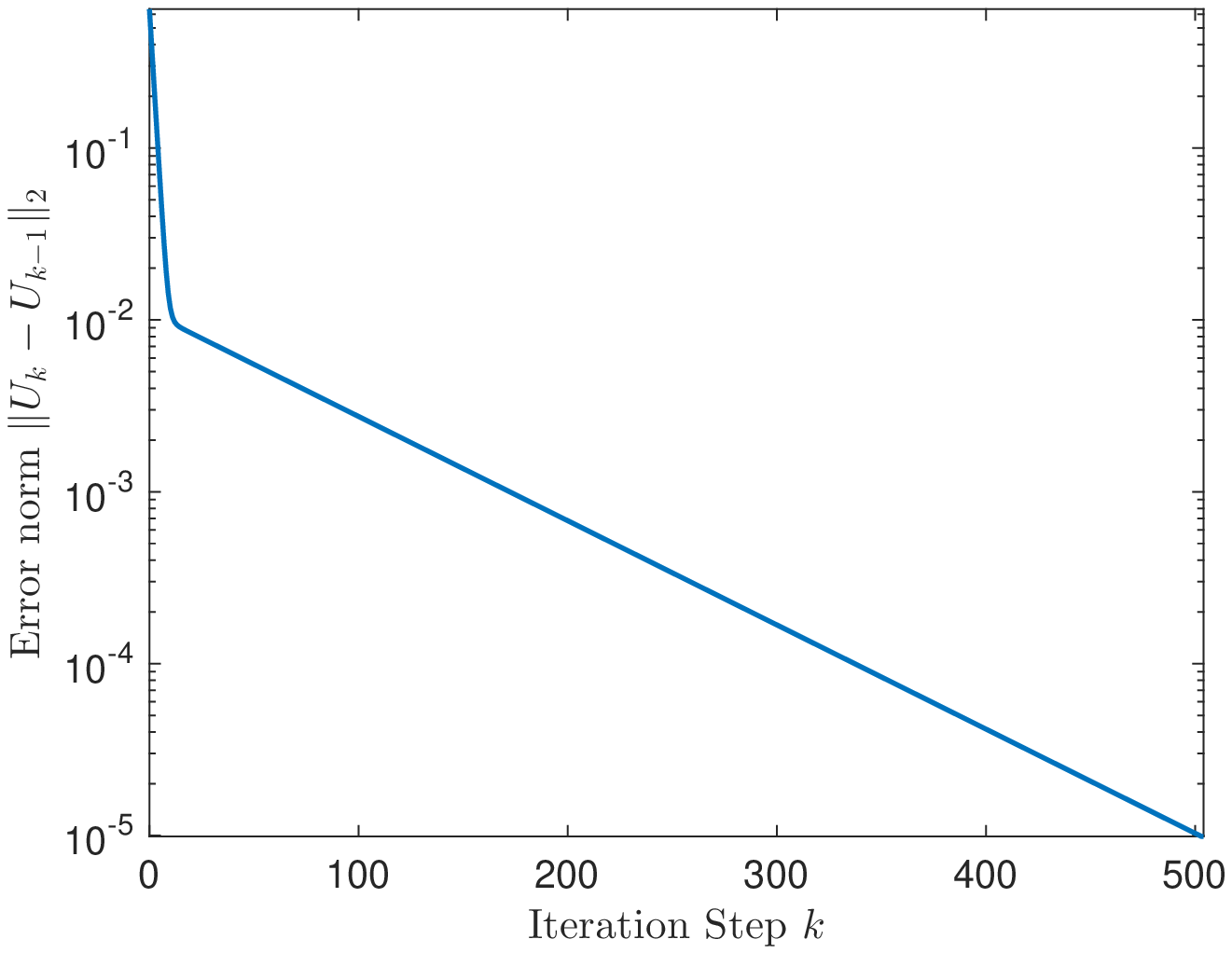}
\caption{(Example 1). Monotonic convergence of the error norm $\left\|\bm{U}_{k}-\bm{U}_{k-1}\right\|_{2}$ along the iteration axis. Left: Case1). Right: Case 2).}\label{fig4}
\end{figure}

{\it Case 2): $\bm{Y}_d=[1,1,2,2]^{\tp}$.}

Based on the application of the Algorithm 1, it follows that $\left\|\bm{U}_{k}-\bm{U}_{k-1}\right\|_{2}<10^{-5}$ occurs after $k=504$ iterations, and also it leads to
\[\bm{U}_{504}=\left[\frac{\Ds547}{\Ds758},\frac{\Ds383}{\Ds2489},\frac{\Ds252}{\Ds1811},-\frac{\Ds440}{\Ds5573},-\frac{\Ds205}{\Ds1258}\right]^{\tp}.
\]

\noindent By simple calculations, we can easily derive $\left\|\bm{Y}_d-\bm{G}\bm{U}_{504}\right\|_{2}\approx1.7321$, owing to which $\bm{U}_{d}\approx\bm{U}_{504}$ can be regarded as a least squares solution for the LAE \eqref{eq1}. This actually coincides with the unsolvability of the LAE (\ref{eq1}) because $\rank\left(\left[\bm{G}~\bm{Y}_d\right]\right)=4$ is greater than $\rank\left(\bm{G}\right)=3$. In addition, $\left\|\bm{Y}_d-\bm{G}\bm{U}_{504}\right\|_{2}$ is close to the least squares norm: $\min_{\bm{\Omega}\in\mathbb{R}^{q}}\left\|\bm{Y}_{d}-\bm{G}\bm{\Omega}\right\|_{2}=1351/780$.

{\it Discussions:} The results developed for two cases verify the effectiveness of the Algorithm 1, regardless of either solvable or unsolvable LAEs. We also depict the evolution of the error norm $\|\bm{U}_{k}-\bm{U}_{k-1}\|_{2}$ along the iteration axis in Fig. \ref{fig4}. Clearly, we can see that the monotonic convergence is achieved for the iterative algorithm (\ref{eq4}) applied to both solvable and unsolvable cases of the LAE (\ref{eq1}). This validates the result in Lemma \ref{lem12}.

\section{Finite-Iteration Convergence Design and Analysis for Solving LAEs}\label{sec4}

In this section, we proceed further to discuss how to improve the exponential convergence results for the iterative algorithm (\ref{eq4}) in Theorems \ref{thm1}-\ref{thm3} by accelerating its convergence rate. Thus, we try to investigate properties of nilpotent matrices such that we can achieve the finite-iteration convergence for the iterative algorithm (\ref{eq4}). We first provide a helpful lemma with properties of nilpotent matrices.

\begin{lem}\label{lem4}
For the iterative algorithm (\ref{eq4}), the following are equivalent design conditions about its gain matrix $\bm{F}\in\mathbb{R}^{q\times p}$.
\begin{enumerate}
\item
The matrix $I_{m}-\widehat{\bm{G}}\bm{F}\bm{H}$ is nilpotent.

\item
There exists some integer $\nu$ ($1\leq\nu\leq m$) such that
\begin{equation}\label{eq43}
\left(I_{m}-\widehat{\bm{G}}\bm{F}\bm{H}\right)^{\nu-1}\neq0~\hbox{and}~
\left(I_{m}-\widehat{\bm{G}}\bm{F}\bm{H}\right)^{\nu}=0.
\end{equation}

\item
The spectral radius of $I_{m}-\widehat{\bm{G}}\bm{F}\bm{H}$ satisfies
\begin{equation}\label{eq58}
\rho\left(I_{m}-\widehat{\bm{G}}\bm{F}\bm{H}\right)=0.
\end{equation}
\end{enumerate}
\end{lem}

\begin{proof}
The equivalence between statements 1) and 2) is a direct consequence of the definition of nilpotent matrices (see, e.g., \cite[Subsection 1.4]{lt:85}), where $\nu$ is the least positive integer satisfying $\left(I_{m}-\widehat{\bm{G}}\bm{F}\bm{H}\right)^{i}=0$ or the degree of nilpotency of the matrix $I_{m}-\widehat{\bm{G}}\bm{F}\bm{H}$. For the equivalence between statements 1) and 3), see, e.g., Exercise 8 of \cite[Subsection 4.11]{lt:85}.
\end{proof}

A fact worth noting for Lemma \ref{lem4} is that there always exists some matrix $\bm{F}\in\mathbb{R}^{q\times p}$ to fulfill its three equivalent conditions, where (\ref{eq38}) taking $\widetilde{\bm{F}}$ as any nilpotent matrix gives a candidate selection of $\bm{F}$. Then with Lemma \ref{lem4}, we show a design result to accomplish the finite-iteration convergence for the iterative solution algorithm (\ref{eq4}) of the LAE (\ref{eq1}).

\begin{thm}\label{thm4}
For the LAE (\ref{eq1}) with any given $\bm{Y}_{d}\in\mathbb{R}^{p}$, let the iterative algorithm (\ref{eq4}) of the selection property (P) be applied. Then there exist some positive integer $\nu$ ($1\leq\nu\leq m$) satisfying (\ref{eq43}) and some $\bm{U}_{\nu}\in\mathbb{R}^{q}$ belonging to $\mathcal{U}_{\mathrm{ILC}}(\bm{Y}_{d})$ defined by (\ref{eq20}) such that
\begin{equation}\label{e42}
\aligned
\bm{U}_{k}&
\left\{\aligned
&\neq\bm{U}_{\nu}, &\forall k&\leq\nu-1\\
&=\bm{U}_{\nu}, &\forall k&\geq\nu
\endaligned
\right.\\
\bm{Y}_{k}&
\left\{\aligned
&\neq\bm{H}\left(\bm{H}^{\tp}\bm{H}\right)^{-1}\bm{H}^{\tp}\bm{Y}_{d}, &\forall k&\leq\nu-1\\
&=\bm{H}\left(\bm{H}^{\tp}\bm{H}\right)^{-1}\bm{H}^{\tp}\bm{Y}_{d}, &\forall k&\geq\nu
\endaligned
\right.
\endaligned
\end{equation}

\noindent if and only if there exists some $\bm{F}\in\mathbb{R}^{q\times p}$ such that (\ref{eq58}) holds. Moreover, the converged values of $\bm{U}_{k}$, $\forall k\in\mathbb{Z}_{+}$ under different selections of the initial condition $\bm{U}_{0}$ can constitute the defined set $\mathcal{U}_{\mathrm{ILC}}(\bm{Y}_{d})$ in (\ref{eq20}) such that
\[
\mathcal{U}_{\mathrm{ILC}}(\bm{Y}_{d})=
\left\{\aligned
&\mathcal{U}_{d}(\bm{Y}_{d}),&&\hbox{if the LAE (\ref{eq1}) is solvable}\\
&\overline{\mathcal{U}}_{d}\left(\bm{Y}_{d}\right),&&\hbox{otherwise}
\endaligned
\right.
\]

\noindent where $\mathcal{U}_{d}(\bm{Y}_{d})$ and $\overline{\mathcal{U}}_{d}\left(\bm{Y}_{d}\right)$ are the solution set in (\ref{eq10}) and the least squares solution set in (\ref{eq27}) for the LAE (\ref{eq1}), respectively.
\end{thm}

\begin{proof}
In view of the proof of Theorem \ref{thm3}, we can leverage the property (P) to derive (\ref{eq22}), and consequently arrive at (\ref{eq32}). Hence, for (\ref{eq32}), we denote $\widehat{\bm{U}}_{1,k}=\widehat{\bm{G}}\bm{U}_{k}$ and $\widehat{\bm{U}}_{2,k}=\widehat{\bm{M}}\bm{U}_{k}$ such that $\widehat{\bm{U}}_{k}=\left[\widehat{\bm{U}}_{1,k}^{\tp},\widehat{\bm{U}}_{2,k}^{\tp}\right]^{\tp}$, and then it is straightforward to obtain from (\ref{eq32}) that
\begin{equation}\label{eq33}
\aligned
\widehat{\bm{U}}_{1,k+1}
&=\left(I_{m}-\widehat{\bm{G}}\bm{F}\bm{H}\right)\widehat{\bm{U}}_{1,k}\\
&~~~+\widehat{\bm{G}}\bm{F}\bm{H}\left(\bm{H}^{\tp}\bm{H}\right)^{-1}\bm{H}^{\tp}\bm{Y}_{d},\quad\forall k\in\mathbb{Z}_{+}
\endaligned
\end{equation}

\noindent and that
\begin{equation}\label{eq34}
\widehat{\bm{U}}_{2,k+1}
=\widehat{\bm{U}}_{2,k},\quad\forall k\in\mathbb{Z}_{+}.
\end{equation}

\noindent Since we can rewrite (\ref{eq33}) as
\begin{equation*}\label{}
\aligned
\widehat{\bm{U}}_{1,k+1}-\left(\bm{H}^{\tp}\bm{H}\right)^{-1}\bm{H}^{\tp}\bm{Y}_{d}
&=\left(I_{m}-\widehat{\bm{G}}\bm{F}\bm{H}\right)\Big[\widehat{\bm{U}}_{1,k}
-\left(\bm{H}^{\tp}\bm{H}\right)^{-1}\bm{H}^{\tp}\bm{Y}_{d}\Big],\quad\forall k\in\mathbb{Z}_{+}
\endaligned
\end{equation*}

\noindent we can easily conclude that there follows
\begin{equation}\label{eq35}
\widehat{\bm{U}}_{1,k}
\left\{\aligned
&\neq\left(\bm{H}^{\tp}\bm{H}\right)^{-1}\bm{H}^{\tp}\bm{Y}_{d}, &\forall k&\leq\nu-1\\
&=\left(\bm{H}^{\tp}\bm{H}\right)^{-1}\bm{H}^{\tp}\bm{Y}_{d}, &\forall k&\geq\nu
\endaligned
\right.
\end{equation}

\noindent if and only if (\ref{eq43}) holds. Based on Lemma \ref{lem4}, we can establish the equivalence between (\ref{eq58}) and (\ref{eq35}).

By applying $\widehat{\bm{U}}_{2,k}=\widehat{\bm{M}}\bm{U}_{k}$ and $\widehat{\bm{N}}\widehat{\bm{M}}=I_{q}-\bm{F}\bm{H}\left(\widehat{\bm{G}}\bm{F}\bm{H}\right)^{-1}\widehat{\bm{G}}$ (see Remark \ref{rem7}), we can employ (\ref{eq34}) to deduce
\begin{equation}\label{eq60}
\aligned
\bm{U}_{k}
&=\widehat{\bm{P}}\widehat{\bm{U}}_{k}\\
&=\bm{F}\bm{H}\left(\widehat{\bm{G}}\bm{F}\bm{H}\right)^{-1}\widehat{\bm{U}}_{1,k}
+\widehat{\bm{N}}\widehat{\bm{U}}_{2,k}\\
&=\widehat{\bm{N}}\widehat{\bm{M}}\bm{U}_{0}
+\bm{F}\bm{H}\left(\widehat{\bm{G}}\bm{F}\bm{H}\right)^{-1}\widehat{\bm{U}}_{1,k}\\
&=\left[I_{q}-\bm{F}\bm{H}\left(\widehat{\bm{G}}\bm{F}\bm{H}\right)^{-1}\widehat{\bm{G}}\right]\bm{U}_{0}\\
&~~~+\bm{F}\bm{H}\left(\widehat{\bm{G}}\bm{F}\bm{H}\right)^{-1}\widehat{\bm{U}}_{1,k},\quad\forall k\in\mathbb{Z}_{+}
\endaligned
\end{equation}

\noindent and therefore, we can conclude from the equivalence between (\ref{eq58}) and (\ref{eq35}) that (\ref{eq58}) also provides a necessary and sufficient condition to arrive at
\[
\aligned
\bm{U}_{k}
&\left\{\aligned
&\neq\bm{U}_{\nu}, &\forall k&\leq\nu-1\\
&=\bm{U}_{\nu}, &\forall k&\geq\nu
\endaligned
\right.
\endaligned
\]

\noindent where the use of (\ref{eq35}) and (\ref{eq60}) with $\widehat{\bm{G}}=\left(\bm{H}^{\tp}\bm{H}\right)^{-1}\bm{H}^{\tp}\bm{G}$ yields
\[\aligned
\bm{U}_{\nu}
&=\left[I_{q}-\bm{F}\bm{H}\left(\widehat{\bm{G}}\bm{F}\bm{H}\right)^{-1}\widehat{\bm{G}}\right]\bm{U}_{0}\\
&~~~+\bm{F}\bm{H}\left(\widehat{\bm{G}}\bm{F}\bm{H}\right)^{-1}\left(\bm{H}^{\tp}\bm{H}\right)^{-1}\bm{H}^{\tp}\bm{Y}_{d}\\
&=\left[I_{q}-\bm{F}\bm{H}\left(\bm{H}^{\tp}\bm{G}\bm{F}\bm{H}\right)^{-1}\bm{H}^{\tp}\bm{G}\right]\bm{U}_{0}\\
&~~~+\bm{F}\bm{H}\left(\bm{H}^{\tp}\bm{G}\bm{F}\bm{H}\right)^{-1}\bm{H}^{\tp}\bm{Y}_{d},\quad\forall\bm{Y}_{d}\in\mathbb{R}^{p}.
\endaligned\]

\noindent Clearly, by (\ref{eq20}), $\bm{U}_{\nu}\in\mathcal{U}_{\mathrm{ILC}}(\bm{Y}_{d})$ holds. Thanks to $\bm{Y}_{k}=\bm{G}\bm{U}_{k}=\bm{H}\widehat{\bm{G}}\bm{U}_{k}=\bm{H}\widehat{\bm{U}}_{1,k}$, we can use (\ref{eq35}) to obtain
\begin{equation*}\label{}
\bm{Y}_{k}
\left\{\aligned
&\neq\bm{H}\left(\bm{H}^{\tp}\bm{H}\right)^{-1}\bm{H}^{\tp}\bm{Y}_{d}, &\forall k&\leq\nu-1\\
&=\bm{H}\left(\bm{H}^{\tp}\bm{H}\right)^{-1}\bm{H}^{\tp}\bm{Y}_{d}, &\forall k&\geq\nu
\endaligned
\right.
\end{equation*}

\noindent which, conversely, also yields \eqref{eq35} by the full-column rank of $\bm{H}$. Consequently, we can derive that (\ref{eq58}) provides a necessary and sufficient condition for the finite-iteration convergence result in (\ref{e42}). In addition, we can directly deduce the remaining results of this theorem by resorting to the results of Theorem \ref{thm3}, which thus are not proved in details here.
\end{proof}

\begin{rem}\label{rem6}
By Theorem \ref{thm4}, we show that the finite-iteration convergence of the iterative algorithm \eqref{eq4} can be achieved. It is thanks to the incorporation of the idea of deadbeat control. The convergence speed of the iterative algorithm \eqref{eq4} depends on the degree $\nu$ of the minimal polynomial of $I_{m}-\widehat{\bm{G}}\bm{F}\bm{H}$ regardless of the solvability of the LAE \eqref{eq1}, namely, $\bm{U}_k$ converges to $\bm{U}_{\nu}$ in the $\nu$th iteration step. Since $\nu\leq m$, the (least squares) solutions for the LAE \eqref{eq1} can be determined within at most $m$ iterations by the application of the iterative algorithm (\ref{eq4}). Additionally, Theorem \ref{thm4} is developed in a unified way for both solvable and unsolvable LAEs by resorting to the property (P). As shown in Remark \ref{rem10} and Lemma \ref{lem11}, this property is actually needed only when using the iterative algorithm (\ref{eq4}) for unsolvable LAEs.
\end{rem}
\begin{rem}\label{rem9}
In Theorem \ref{thm4}, our iterative method is designed such that the solving of LAEs can be realized in finite-iteration steps. In contrast to the common existing iterative methods that generally converge asymptotically, our iterative method gets a faster convergence rate and operates more like a direct method able to obtain the exact (least squares) solutions for any LAEs within finite operations, in addition to the robustness property and predictable behaviors of iterative methods. It is particularly significant for practical applications, especially for those under the high accuracy requirements. Thus, the ideas in the control may provide a feasible way to improve the design of iterative algorithms for solving the mathematical problems, particularly for improving the convergence rate.  This, however, may not be easily realized by resorting to the popular iterative methods in the mathematics. As a consequence, it discloses that although the tools from the mathematics provide the basis for the control design and analysis, the control design and analysis tools may, in turn, be beneficial to enhance the problem-solving methods in the mathematics.
\end{rem}

\section{Applications to ILC Design and Analysis}\label{sec5}

In this section, we show how to employ the lifting technique to transform the perfect tracking problem for discrete-time ILC systems into the solving problem for LAEs. Then we apply the observer-based design methods of solving LAEs to accomplish the design and analysis of conventional 2-D ILC systems with evolution along not only the infinite iteration axis, denoted by $k\in\mathbb{Z}_{+}$, but also a finite time axis, denoted by $t\in\mathbb{Z}_{N}$.

We specifically consider ILC for a discrete-time system as
\begin{equation}\label{eq42}
\left\{\aligned
x_{k}(t+1)
&=Ax_{k}(t)+Bu_{k}(t)+w(t)\\
y_{k}(t)
&=Cx_{k}(t)+v(t)
\endaligned\right.,\quad\forall t\in\mathbb{Z}_{N},\forall k\in\mathbb{Z}_{+}
\end{equation}

\noindent where $x_{k}(t)\in\mathbb{R}^{n_{s}}$, $u_{k}(t)\in\mathbb{R}^{n_{i}}$ and $y_{k}(t)\in\mathbb{R}^{n_{o}}$ are the system state, input and output, respectively; $A$, $B$ and $C$ are the system matrices with required dimensions; and $w(t)\in\mathbb{R}^{n_{s}}$ and $v(t)\in\mathbb{R}^{n_{o}}$ are the iteration-invariant disturbances. To implement ILC for the system \eqref{eq42}, we adopt the iteration-invariant initial state as $x_{k}(0)=x_{0}\in\mathbb{R}^{n_{s}}$, $\forall k\in\mathbb{Z}_{+}$, and take any initial input $u_0(t)\in\mathbb{R}^{n_{i}}$, $\forall t\in\mathbb{Z}_{N-1}$. Without loss of generality, we assume that the relative degree of the system \eqref{eq42} is $r>0$, where $r$ is also the time delay in the output when the influence of the input works \cite{sw:03}. That is, we have
\begin{equation}\label{eq46}
CA^{j}B=0, \forall j=0,1,\cdots,r-2~~\hbox{and}~~CA^{r-1}B\neq0.
\end{equation}

\noindent Furthermore, we denote the Markov parameter matrices of the system \eqref{eq42} as ${G}_i=CA^{i+r-1}B$, $\forall i\in\mathbb{Z}_{+}$, where $G_0$ is the first nonzero Markov parameter matrix.

In general, a problem of ILC is to find some input sequence, under which the output of the system \eqref{eq42} can track the desired reference for all time steps within the finite interval of interest after an iterative process. To be specific, in view of the relative degree $r$, the perfect tracking task is to determine some proper control input sequence $\left\{u_k(t):k\in\mathbb{Z}_{+},t\in\mathbb{Z}_{N-1}\right\}$, converging to some desired input $u_d(t)$, $\forall t\in\mathbb{Z}_{N-1}$, such that
\begin{equation}\label{eq47}
\lim_{k\to\infty}y_{k}(t)=y_{d}(t),\quad\forall t=r,r+1,\cdots,r+N-1
\end{equation}

\noindent where $y_d(t)\in\mathbb{R}^{n_o}$ denotes the desired reference.

To proceed, we consider the lifting technique and introduce six supervectors in the form of
\begin{equation}\label{eq48}
\aligned
\bm{U}_{k}
&=\left[u_{k}^{\tp}(0),u_{k}^{\tp}(1),\cdots,u_{k}^{\tp}(N-1)\right]^{\tp},~~\forall k\in\mathbb{Z}_{+}\\
\bm{U}_{d}
&=\left[u_{d}^{\tp}(0),u_{d}^{\tp}(1),\cdots,u_{d}^{\tp}(N-1)\right]^{\tp}\\
\bm{W}
&=\left[w^{\tp}(0),w^{\tp}(1),\cdots,w^{\tp}(N+r-2)\right]^{\tp}\\
\bm{Y}_{k}
&=\left[y_{k}^{\tp}(r),y_{k}^{\tp}(r+1),\cdots,y_{k}^{\tp}(N+r-1)\right]^{\tp},~~\forall k\in\mathbb{Z}_{+}\\
\bm{Y}_{d}
&=\left[y_{d}^{\tp}(r),y_{d}^{\tp}(r+1),\cdots,y_{d}^{\tp}(N+r-1)\right]^{\tp}\\
\bm{V}
&=\left[v^{\tp}(r),v^{\tp}(r+1),\cdots,v^{\tp}(N+r-1)\right]^{\tp}.
\endaligned
\end{equation}

\noindent Thus, we can write the system \eqref{eq42} with supervectors as
\begin{equation}\label{e1}
\bm{Y}_k=\bm{G}\bm{U}_k+\bm{X}_0+\bm{D}\bm{W}+\bm{V}
\end{equation}

\noindent where
\begin{equation}\label{eq49}
{\small
\aligned
\bm{G}
&=\begin{bmatrix}
G_{0}&0&\cdots&0\\
G_{1}&G_{0}&\ddots&\vdots\\
\vdots&\vdots&\ddots&0\\
G_{N-1}&G_{N-2}&\cdots&G_{0}
\end{bmatrix},\quad
\bm{X}_0=\begin{bmatrix}
CA^{r}\\
CA^{r+1}\\
\vdots\\
CA^{r+N-1}\end{bmatrix}x_0\\
\bm{D}
&=\begin{bmatrix}
CA^{r-1}&CA^{r-2}&\cdots&C&0&\cdots&0\\
CA^{r}&CA^{r-1}&\vdots&CA&C&\ddots&\vdots\\
\vdots&\vdots&\vdots&\vdots&\vdots&\ddots&0\\
CA^{r+N-2}&CA^{r+N-3}&\cdots&CA^{N-1}&CA^{N-2}&\cdots&C
\end{bmatrix}.
\endaligned
}
\end{equation}

\noindent Correspondingly, the tracking objective (\ref{eq47}) is equivalent to
\begin{equation}\label{eq50}
\lim_{k\to\infty}\bm{Y}_k=\bm{Y}_d.
\end{equation}

\noindent If we denote $\widetilde{\bm{Y}}_d=\bm{Y}_d-(\bm{X}_0+\bm{D}\bm{W}+\bm{V})$, then we can easily see that the iterative seeking of $\bm{U}_{k}$ for (\ref{e1}) to reach the objective (\ref{eq50}) coincides with the iterative solving of an LAE given by
\begin{equation}\label{e2}
\widetilde{\bm{Y}}_{d}=\bm{G}\bm{U}_{d}.
\end{equation}

\noindent In particular, this clearly explains the relation between designing ILC and solving LAEs, as exemplified in Subsection \ref{sec11}.

Motivated by the abovementioned analysis, we benefit from Theorems \ref{thm3} and \ref{thm4} to develop a tracking result of ILC for the system \eqref{eq42} in the following theorem.

\begin{thm}\label{thm5}
For any desired reference $y_{d}(t)$, the system (\ref{eq42}) can achieve the tracking objective (\ref{eq47}) under an updating law of ILC involving the gain matrices $F_{ij}\in\mathbb{R}^{n_{i}\times n_{o}}$, $\forall i$, $j=1$, $2$, $\cdots$, $N$ and taking the form of
\begin{equation}\label{eq51}
\aligned
u_{k+1}(t)&=u_{k}(t)+\sum_{i=0}^{N-1}F_{t+1,i+1}\big[y_{d}(i+r)
-y_{k}(i+r)\big],\quad\forall t\in\mathbb{Z}_{N-1},\forall k\in\mathbb{Z}_{+}
\endaligned
\end{equation}

\noindent if and only if the LAE (\ref{e2}) is solvable. Moreover, the design condition of (\ref{eq51}) is that its gain matrices can make $\bm{F}=\left[F_{ij}\right]\in\mathbb{R}^{Nn_{i}\times Nn_{o}}$ to satisfy the spectral radius condition (\ref{eq39}) or (\ref{eq58}).
\end{thm}

\begin{proof}
By the supervectors in (\ref{eq48}), we can write (\ref{eq51}) as
\[
\bm{U}_{k+1}
=\bm{U}_{k}+\bm{F}\left(\bm{Y}_{d}-\bm{Y}_{k}\right)
\]

\noindent which, together with (\ref{e1}), further leads to
\begin{equation}\label{e8}
\bm{U}_{k+1}
=\left(I_{n_{i}N}-\bm{F}\bm{G}\right)\bm{U}_{k}+\bm{F}\widetilde{\bm{Y}}_{d},\quad\forall k\in\mathbb{Z}_{+}.
\end{equation}

\noindent Clearly, (\ref{e8}) that has the same structure as (\ref{eq4}) gives an iterative solution algorithm for the LAE (\ref{e2}). Hence, by the application of Theorems \ref{thm3} and \ref{thm4} to (\ref{e8}), we know that either the condition (\ref{eq39}) or (\ref{eq58}) can ensure the convergence of $\bm{U}_{k}$. Consequently, we can further develop the equivalence between the solvability of the LAE (\ref{e2}) and the accomplishment of (\ref{eq47}) for the system (\ref{eq42}) under the ILC updating law (\ref{eq51}).
\end{proof}

\begin{rem}\label{rem12}
In Theorem \ref{thm5}, the ILC tracking task is realized by bridging an equivalence relation between it and the solving problem of an LAE. This provides us a way to newly introduce an observer-based design method of ILC that yields algorithms beyond the typical algorithm framework for ILC, especially for PID-type ILC. Furthermore, our observer-based design method of ILC effectively works even without the full (row or column) rank condition of $G_{0}$, which is generally considered as one of the most fundamental assumptions for ILC (see, e.g., \cite{bta:06,acm:07}). As another benefit, our observer-based design method provides a potential way to narrow the gap between the classical design methods of ILC and the popular design methods based on state observers, which makes it possible to improve the ILC design and analysis with feedback-based tools. A direct consequence of this benefit is that when the spectral radius condition \eqref{eq58} is ensured, the perfect tracking objective of ILC can be achieved only within finite iterations.
\end{rem}

For the most commonly considered cases when the full rank condition of $G_0$ is satisfied, we may employ Theorems \ref{thm1} and \ref{thm2} to develop further tracking results of ILC for the system (\ref{eq42}). We first establish a corollary for ILC of the system (\ref{eq42}) under the full-column rank condition on $G_{0}$.

\begin{cor}\label{cor2}
Consider the system (\ref{eq42}) with $\rank(G_{0})=n_{i}$. If the updating law (\ref{eq51}) of ILC is applied with its gain matrices selected such that $\bm{F}$ satisfies the spectral radius condition (\ref{eq5}), then for any desired reference $y_{d}(t)$, the tracking objective (\ref{eq47}) can be achieved, together with yielding a unique learned input $\lim_{k\to\infty}u_{k}(t)=u_{d}(t)$, $\forall t\in\mathbb{Z}_{N-1}$, if and only if the LAE (\ref{e2}) is solvable. In particular, if the updating law (\ref{eq51}) takes a simple form of
\begin{equation}\label{eq52}
u_{k+1}(t)=u_{k}(t)+F_{0}\left[y_{d}(t+r)-y_{k}(t+r)\right],~\forall t\in\mathbb{Z}_{N-1},\forall k\in\mathbb{Z}_{+}
\end{equation}

\noindent for some gain matrix $F_{0}\in\mathbb{R}^{n_{i}\times n_{o}}$, then $\bm{F}=I_N\otimes F_{0}$ holds, and the spectral radius condition \eqref{eq5} is equivalent to
\[
\rho\left(I_{n_{i}}-F_{0}G_{0}\right)<1.
\]
\end{cor}

\begin{proof}
Owing to the equivalence between the full-column rank of $\bm{G}$ and that of $G_0$, this corollary can be obtained based on Theorem \ref{thm1} by directly following the same lines as the proof of Theorem \ref{thm5}, and thus its proof is omitted here.
\end{proof}

For the case when ${G}_0$ is with the full-row rank, the following corollary reveals that the tracking objective \eqref{eq47} for the system \eqref{eq42} can be realized in the presence of any desired reference.

\begin{cor}\label{cor3}
For the system \eqref{eq42} with any desired reference $y_{d}(t)$, the tracking objective (\ref{eq47}) can always be achieved under the updating law \eqref{eq51} of ILC if and only if any of the following conditions holds:
\begin{enumerate}
\item
the LAE (\ref{e2}) is always solvable;

\item
$\rank(G_{0})=n_{o}$ holds;

\item
there exist some gain matrices of (\ref{eq51}) such that $\bm{F}$ fulfills the spectral radius condition (\ref{eq8}).
\end{enumerate}

\noindent In particular, when applying the updating law (\ref{eq52}), the spectral radius condition (\ref{eq8}) equivalently becomes
\[
\rho\left(I_{n_{i}}-G_{0}F_{0}\right)<1.
\]
\end{cor}

\begin{proof}
Let $\widetilde{\bm{Y}}_{k}=\bm{Y}_{k}-(\bm{X}_0+\bm{D}\bm{W}+\bm{V})$. Thus, we can write (\ref{e1}) as $\widetilde{\bm{Y}}_{k}=\bm{G}\bm{U}_{k}$, and (\ref{eq47}) (or equivalently (\ref{eq50})) holds if and only if $\lim_{k\to\infty}\widetilde{\bm{Y}}_{k}=\widetilde{\bm{Y}}_{d}$. With this fact and by the equivalence between the full-row rank of $\bm{G}$ and that of $G_0$, we can develop the proof of this corollary in a similar way as that of Theorem \ref{thm5} by resorting to Lemma \ref{lem2} and Theorem \ref{thm2}.
\end{proof}

In Corollaries \ref{cor2} and \ref{cor3}, we employ our observer-based design method to particularly give some tracking results of ILC in the presence of the full rank condition of $G_{0}$. It is clear to see that the widely used P-type updating law (\ref{eq52}) in ILC is included as a special case of our observer-based updating law (\ref{eq51}). Thanks to Theorem \ref{thm5} and Corollary \ref{cor3}, we can also note an interesting fact that to accomplish the same output tracking task, there can emerge different inputs learned via the ILC process, depending on the different selections of the initial input. It thus implies that the commonly made realizability hypothesis, requiring the uniqueness of the input to realize the specified output tracking task in the ILC literature (see, e.g., \cite{sw:02,xypy:16,mm:171,szwc:16}), is not necessary for the perfect tracking tasks of ILC, thanks to the application of the observer-based design method of ILC.

Next, we use a simple example to illustrate the effectiveness of our observer-based design method of ILC without requiring the full rank condition of the first nonzero Markov parameter matrix $G_{0}$.

{\it Example 2:} Consider the system \eqref{eq42} whose system matrices are given by
\[
A=\begin{bmatrix}
1&1&0\\
0&1&1\\
0&0&1
\end{bmatrix},\quad
B=\begin{bmatrix}
1&-1\\
2&-2\\
0&0
\end{bmatrix},\quad
C=\begin{bmatrix}
1&0\\
0&1\\
1&-1
\end{bmatrix}^{\tp}
\]

\noindent and thus, (\ref{eq42}) is an unstable system \cite{am:06}. We can easily obtain
\[G_0=CB=\begin{bmatrix}
1&-1\\
2&-2
\end{bmatrix}\]

\noindent such that the relative degree $r=1$. It is clear that $G_0$ is neither of full-column rank nor of full-row rank. Additionally, we set the time period as $N=30$, consider disturbances as
\[
w(t)=\left[0,0,0\right]^{\tp},\quad
v(t)=\left[2\sin(0.2t+1)-1,2\sin(0.2t)\right]^{\tp}
\]

\noindent and adopt the desired reference as
\[
y_d(t)=\left[2\sin(0.2t+1),2\sin(0.2t)\right]^{\tp}.
\]

\noindent To achieve the tracking objective (\ref{eq47}) for the system (\ref{eq42}), we employ the updating law (\ref{eq51}) by adopting the initial conditions with $x_0=[1,0,0]^{\tp}$ and $u_0(t)=[5,1]^{\tp}$, $\forall t\in\mathbb{Z}_{29}$.

Even though $\bm{G}$ is neither of full-column rank nor of full-row rank, we can verify that the LAE \eqref{e2} is solvable. By Theorem \ref{thm5}, the perfect tracking objective \eqref{eq47} can be achieved under the updating law (\ref{eq51}). To use this updating law, we next determine its gain matrices, for which we decompose $\bm{G}=\bm{H}\widehat{\bm{G}}$ such that
\[
\widehat{\bm{G}}=\begin{bmatrix}
1&-1&0&0&\cdots&0&0\\
0&0&1&-1&\ddots&\vdots&\vdots\\
\vdots&\vdots&\vdots&\vdots&\ddots&0&0\\
0&0&0&0&\cdots&1&-1
\end{bmatrix}\in\mathbb{R}^{30\times60}
\]

\noindent and
\[
\bm{H}=\begin{bmatrix}
1&2&3&2&\cdots&59&2\\
0&0&1&2&\ddots&\vdots&\vdots\\
\vdots&\vdots&\vdots&\vdots&\ddots&3&2\\
0&0&0&0&\cdots&1&2
\end{bmatrix}^{\tp}\in\mathbb{R}^{60\times30}.
\]

\noindent Obviously, we have $\rank(\bm{H})=\rank(\widehat{\bm{G}})=\rank(\bm{G})=30$. From the structured forms of $\widehat{\bm{G}}$ and $\bm{H}$, we denote $\widehat{G}_0=\left[1,-1\right]$ and $H_0=\left[1,2\right]^{\tp}$, and particularly choose $\bm{F}=I_{30}\otimes F_{0}$ with
\[F_{0}=\left[\begin{matrix}2&1\\
1&1
\end{matrix}\right].
\]

\noindent We can thus obtain $1-\widehat{G}_{0}F_{0}H_{0}=0$, which ensures the spectral radius condition (\ref{eq58}). Actually, $\left(I_{30}-\widehat{\bm{G}}\bm{F}\bm{H}\right)^{30}=0$ holds, and as a result, we can realize the perfect tracking task (\ref{eq47}) within $30$ iterations from Theorems \ref{thm4} and \ref{thm5}, i.e., (\ref{eq47}) further becomes $y_{k}(t)=y_{d}(t)$, $\forall t=1$, $2$, $\cdots$, $30$, $\forall k\geq30$.

\begin{figure}
\centering
\includegraphics[width=2.8in]{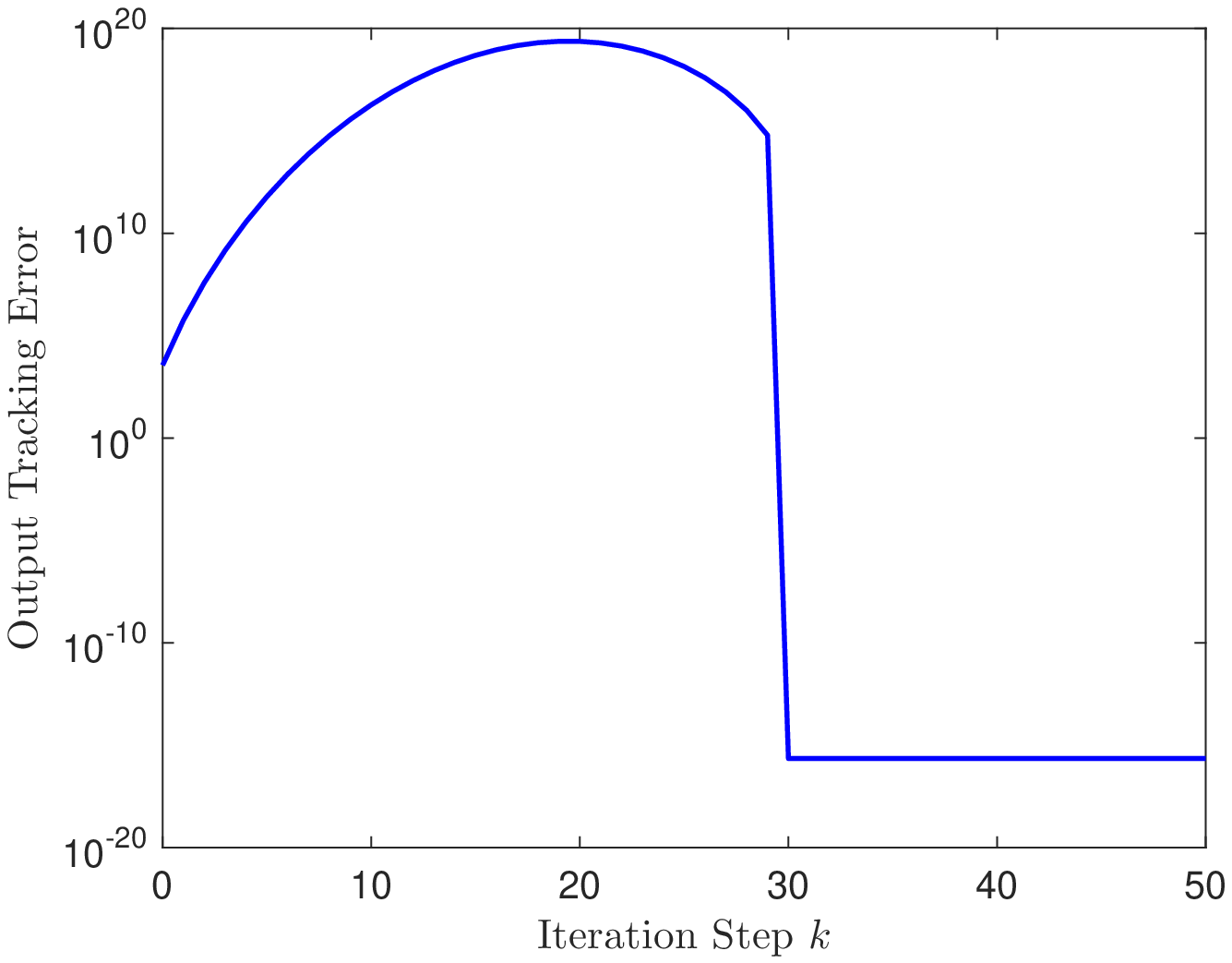}\includegraphics[width=2.8in]{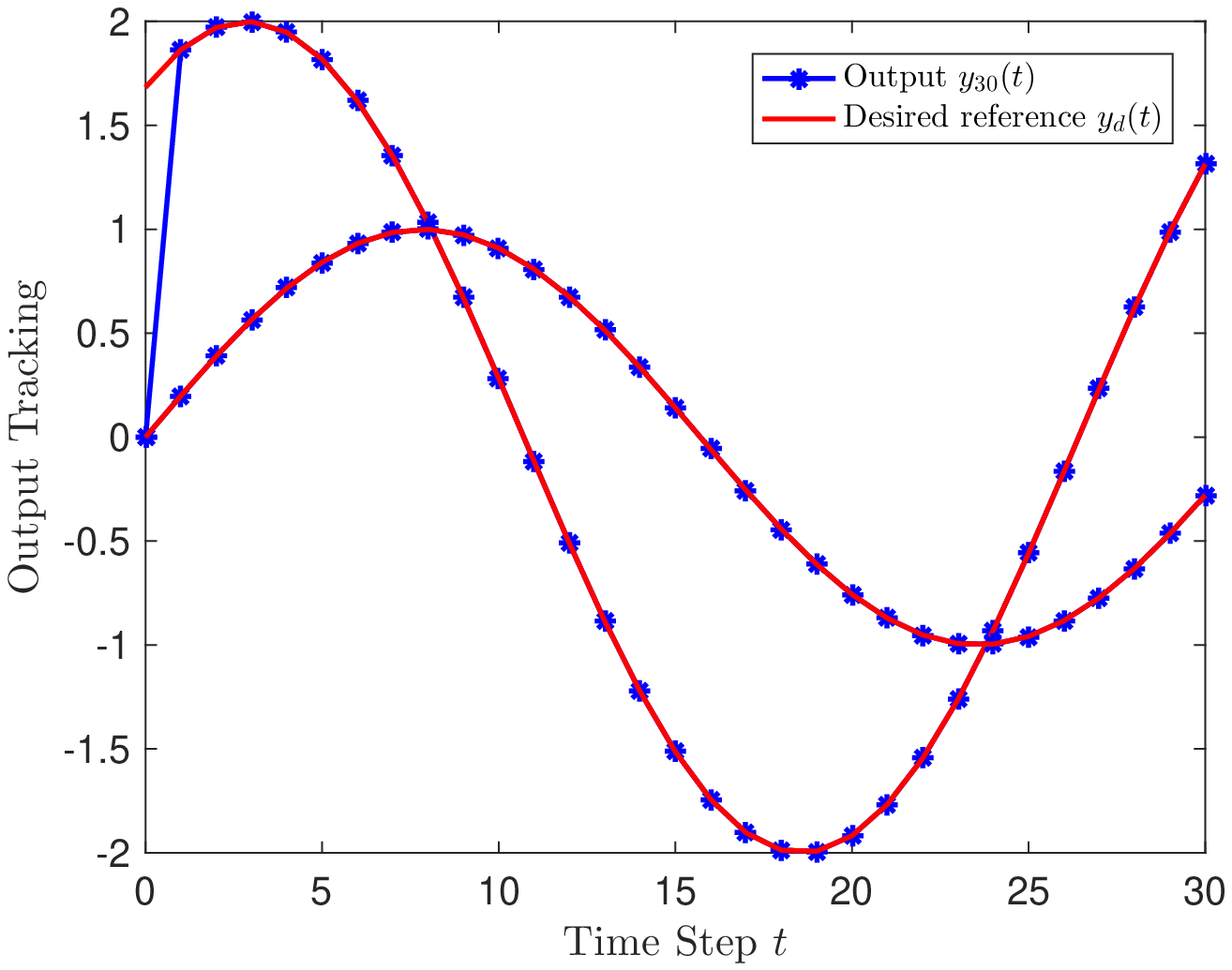}
\caption{(Example 2). The tracking performance of the system (\ref{eq42}) under the ILC updating law (\ref{eq51}). Left: Iteration evolution of the tracking error given by $\max_{1\leq t\leq30}\left\|y_{d}(t)-y_{k}(t)\right\|_{2}$. Right: Output tracking result at the $30$th iteration.}\label{fig5}
\end{figure}

In Fig. \ref{fig5}, we depict the output tracking performance for the system (\ref{eq42}) when peforming simulations with the updating law (\ref{eq51}) of ILC for the first $50$ iterations. By the iteration evolution of the tracking error, evaluated by $\max_{1\leq t\leq30}\left\|y_{d}(t)-y_{k}(t)\right\|_{2}$, we can find that the finite-iteration convergence is achieved for ILC within $30$ iterations in spite of the large learning transient growth. This clearly illustrates the validity of the incorporation of the observer-based design for bettering the ILC convergence speed, which even works in the presence of an unstable system. In particular, the output tracking result given in Fig. \ref{fig5} indicates that the output learned via ILC at the $30$th iteration can exactly follow the desired reference for all time steps except the initial time step. It coincides with the typical output tracking results of ILC in the presence of nonzero system relative degrees \cite{acm:07}.


\section{Conclusions}\label{sec6}

In this paper, how to feed back the ``control design'' idea into the mathematic problems has been discussed, which focuses on facilitating the bidirectional interactions between mathematics and control. Towards this end, an equivalence relation between the solving problem of any LAE and the state observer design problem of its associated control system has been established. It provides guarantee for incorporating the design idea of state observers into the design of a new iterative solution method for the solving problem of LAEs. Consequently, a tight connection of the (least squares) solutions for LAEs with the fundamental observability properties for control systems has been disclosed reasonably. It has been also revealed that all the (least squares) solutions for any LAEs can be determined, depending linearly on the initial condition. These facts may provide new ideas and insights, especially from the perspective of basic observability-related problems, for bettering control-based iterative methods of solving LAEs.

For the convergence analysis of our observer-based iterative solution method, conditions have been established for not only exponential convergence but also monotonic convergence of it, regardless of any LAEs. Moreover, the design idea of deadbeat control has been incorporated for accelerating its convergence speed such that the finite-iteration convergence can be realized. As an application, the observer-based iterative solution method of LAEs has been employed to accomplish the perfect tracking tasks for conventional 2-D ILC systems. Not only can new ILC algorithms falling beyond the typical design framework of ILC be obtained, but also some fundamental hypotheses commonly made in ILC can be removed, such as the full rank assumption relevant to the system relative degree. In addition, our adopted method may provide a feasible way to connect the ILC design with the  popular feedback-based design.


\section*{Acknowledgement}

The author would like to thank Dr. Y. Wu and Dr. J. Zhang, Beihang University (BUAA), for their helpful discussions.

\end{document}